\PassOptionsToPackage{table}{xcolor}
\documentclass{lmcs} 


\usepackage{hyperref}
\usepackage{graphicx}
\usepackage{marvosym}
\usepackage{float}
\usepackage{tikz}
\usetikzlibrary{positioning}
\usetikzlibrary{datavisualization}
\usetikzlibrary{datavisualization.formats.functions}
\usepackage{wrapfig}

\usepackage{mathtools}
\theoremstyle{plain}
\allowdisplaybreaks

\usepackage{xcolor}
\usepackage{amsmath}
\usepackage{mathtools}
\usepackage{suffix}
\usepackage{amsfonts}
\usepackage{mathpartir}
\usepackage{enumitem}
\usepackage{stmaryrd}
\usepackage[all]{xy}
\usepackage{amssymb}
\usepackage{twoopt}

\usepackage{amsthm}

\theoremstyle{definition}

\newtheorem{exampleEnv}{Example}
\newtheorem*{exampleEnv*}{Example}
\newtheorem*{exampleEnv+}{Example \oktheorem@parameter}
\newenvironment{example}{\begin{exampleEnv}}{\qed\end{exampleEnv}}
\newenvironment{example*}{\begin{exampleEnv*}}{\qed\end{exampleEnv*}}
\newenvironment{example+}[1]{\def\oktheorem@parameter{#1}\begin{exampleEnv+}}{\qed\end{exampleEnv+}}

\makeatletter
\newcommand\@TyAlph[1]{\ifcase #1\or \tau\or \sigma\or \rho\else \@ctrerr \fi }
\newcommand\ty[1][1]{{\@TyAlph{#1}}}

\newcommand\tvar[1][1]{{\@TyVarAlph{#1}}}
\newcommand\@TyVarAlph[1]{\ifcase #1\or \alpha\or \beta\or \gamma\else \@ctrerr \fi }

\newcommand\var[1][1]{{\@VarAlph{#1}}}
\newcommand\@VarAlph[1]{\ifcase #1\or x\or y\or z\or u\or v\or w\else \@ctrerr \fi }

\newcommand\trm[1][1]{{\@TermAlph{#1}}}
\newcommand\@TermAlph[1]{\ifcase #1\or t\or s\or r\else \@ctrerr \fi }

\newcommand\val[1][1]{\ifcase #1\or v\or w\or u\else \@ctrerr \fi }

\newcommand\op[1][1]{\ifcase #1\or \mathsf{op}\or \mathsf{op}'\or \mathsf{op}''\else \@ctrerr \fi }
\makeatother

\newcommand\Op{\mathsf{Op}}
\newcommand\cnst{\underline{c}}
\newcommand\sigmoid{\varsigma}

\newcommand\bp[1]{\boldsymbol{(}#1\boldsymbol{)}}

\newcommand\tPair[2]{\langle #1, #2\rangle}
\newcommand\tTriple[3]{\langle #1, #2, #3\rangle}
\newcommand\tTuple[1]{\langle #1\rangle}
\newcommand\tInj[3][\,]{#2.#3#1}
\newcommand\Cns{\ell}
\newcommand\Inj[2][\,]{\mathsf{#2}#1}
\newcommand\tNothingSym{\mathsf{Nothing}}

\newcommand\tJustSym{\mathsf{Just}}

\newcommand\tNil{[\,]}
\newcommand\tCons[2]{#1::#2}

\newcommand\fun[1]{\lambda #1.}
\newcommand\letin[3]{\mathbf{let}\,#1=\,#2\,\mathbf{in}\,#3}

\newcommand\pMatch[5][\,]{\mathbf{case}\,#2\,\mathbf{of}#1\tPair{#3}{#4}\To#5}
\newcommand\tMatch[4][\,]{\mathbf{case}\,#2\,\mathbf{of}#1\tTuple{#3}\To#4}
\newcommand\vMatch[3][\,]{\mathbf{case}\,#2\,\mathbf{of}#1\{#3\}}

\newcommand\lFold[5]{\mathbf{fold}\, (#1,#2).#3\,\mathbf{over}\,#4\,\mathbf{from}\,#5}

\newcommand\ctx{\Gamma}
\newcommand\tinf{\vdash}
\newcommand\Ginf[3][]{\ctx #1\tinf #2 : #3}
\newcommand\subst[2]{#1{}[#2]}
\newcommand\sfor[2]{^{#2}\!/\!_{#1}}

\newcommand\reals{\mathbf{real}}

\newcommand\Unit{\bp{\,}}
\newcommand\Variant[1]{\{ #1 \}}
\WithSuffix\newcommand\t*{\boldsymbol{\mathop{*}}}
\newcommand\MaybeSym{\mathbf{maybe}}
\newcommand\ListSym{\mathbf{list}}

\newcommand\Maybe[1]{\MaybeSym(#1)}
\newcommand\List[1]{\ListSym(#1)}

\newcommand\To{\to}

\newcommand\bProd[2]{\bp{#1 \t* #2}}
\newcommand\tProd[3]{\bp{#1 \t* #2 \t* #3}}

\newcommand\Dsynsymbol[1][]{\scalebox{0.8}{$\overrightarrow{\mathcal{D}}$}_{#1}}\newcommand\Dsyn[2][]{\Dsynsymbol[#1](#2)}

\newcommand\Syn{\mathbf{Syn}}

\newcommand\oframebox[1]{\framebox{#1}}

\newcommand\tFromMaybe[1]{\mathrm{fromMaybe}}\newcommand\tFromMayben[2]{\mathrm{fromMaybe}^{#2}}

\newcommand\tMap[2]{\mathrm{map}}

\newcommand{\plots}[1]{\mathcal{P}_{#1}}

\newcommand\freeeq[1]{\stackrel{\# #1}{=}}

\definecolor{shade}{RGB}{223,223,223}
\definecolor{unshade}{RGB}{255,255,255}

\usepackage{tcolorbox}
\newtcbox{\shadebox}{on line,arc=1pt, outer arc=2pt,colback=shade,colframe=shade,boxsep=0pt,left=1pt,right=1pt,top=2pt,bottom=2pt,boxrule=0pt,bottomrule=1pt,toprule=1pt}

\newtcbox{\unshadebox}{on line,arc=1pt, outer arc=2pt,colback=unshade,colframe=shade,boxsep=0pt,left=1pt,right=1pt,top=2pt,bottom=2pt,boxrule=0pt,bottomrule=1pt,toprule=1pt}

\newcommand\syncat[1]{\mspace{-25mu}\synname{#1}}
\newcommand\synname[1]{\qquad\text{#1}}

\newenvironment{syntax}[1][]{\(
  \rowcolors{100}{white}{white}\begin{array}[t]{#1l@{\quad\!\!}*3{l@{}}@{\,}l}
}{
\end{array}
\)}

\newcommand\gdefinedby{::=}
\newcommand\gor{\mathrel{\lvert}}

\newcommand\vor{\mathrel{\big\lvert}}

\newcommand{\dif}{\mathop{}\!\mathrm{d}}
\newcommand{\Diff}{\mathbf{Diff}}
\newcommand{\CartSp}{\mathbf{CartSp}}
\newcommand{\Man}{\mathbf{Man}}

\newcommand{\Set}{\mathbf{Set}}
\newcommand{\sem}[1]{\llbracket #1\rrbracket}
\newcommand{\semgl}[1]{\llparenthesis #1\rrparenthesis}
\newcommand{\RR}{\mathbb{R}}
\newcommand{\NN}{\mathbb{N}}
\newcommand\catC{\mathcal{C}}

\newcommand\freeF{F}

\newcommand\Dsemsymbol[1][]{\mathcal{J}^{#1}}\newcommand\Dsem[2][]{\Dsemsymbol[#1](#2)}

\newcommand\evRsymbol[1][]{\mathrm{evR}^{#1}}
\newcommandtwoopt\evR[3][][]{\evRsymbol[#2]_{#1}(#3)}
\newcommand\lamRsymbol[1][]{\mathrm{lamR}^{#1}}
\newcommandtwoopt\lamR[3][][]{\lamRsymbol[#2]_{#1}(#3)}

\newcommand{\Gl}[1][]{\mathbf{Gl}_{#1}}

\newcommand{\sPair}[2]{( #1, #2 )}

\newcommand{\sTuple}[1]{(#1)}

\newcommand\id[1][{}]{{\rm id}_{#1}}

\newcommand\projf{\mathrm{proj}}

\newcommand\xto\xrightarrow

\newcommand\DtoT[2][]{\phi_{#2#1}^{\Dsynsymbol\Dsemsymbol}}

\newcommand\ev{\mathrm{ev}}
\newcommand\sort{\mathrm{sort}}
\newcommand\swap{\mathrm{swap}}

\newcommand\innerprod[1]{\cdot_{#1}}

\newcommand\y{\mathbf{y}}

\newcommand{\set}[1]{\left\{#1\right\}}
\newcommand{\cchoose}[2]{{#1\choose #2 }}

\newcommand\seq[2][]{\left(#2\right)_{#1}}
\newcommand\coseq[2][]{\left[#2\right]_{#1}}

\newcommand\cover{\mathcal{U}}

\renewcommand\lim{\mathrm{lim}}

\renewcommand\circ{
  \PackageError{Matthijs}{Please use ; instead! Or if you'd rather use circ consistently, tell me and we'll switch the whole document over}{}}

\newcommand{\defeq}{\stackrel {\mathrm{def}}=}

\usepackage{todonotes}

\begin{document}

\title[Higher-Order Automatic Differentiation of Higher-Order Functions]{Higher-Order Automatic Differentiation of Higher-Order Functions}

\author[M.~Huot]{Mathieu Huot}	\address{University of Oxford}	\thanks{The authors contributed equally to this work.}	

\author[S.~Staton]{Sam Staton}	\address{University of Oxford}	

\author[M.~V\'ak\'ar]{Matthijs V\'ak\'ar}	\address{Utrecht University}	

\begin{abstract}
  \noindent We present semantic correctness proofs of automatic differentiation (AD).
  We consider a forward-mode AD method on a higher-order language with algebraic data types, and we characterise it as the
  unique structure-preserving macro given a choice of derivatives for basic
  operations. We describe a rich semantics for differentiable programming,
  based on diffeological spaces. We show that it interprets our language,
  and we phrase what it means for the AD method to be correct with respect to this semantics.
  We show that our characterisation of AD gives
  rise to an elegant semantic proof of its correctness based on a gluing
  construction on diffeological spaces. We explain how this is, in essence,
  a logical relations argument. Throughout, we show how the analysis extends
  to AD methods for 
  computing higher-order derivatives using a Taylor approximation.
\end{abstract} 
\maketitle

\section{Introduction}
Automatic differentiation (AD), loosely speaking, is the process of taking a program describing a function, and constructing the derivative of that function by applying the chain rule across the program code.
As gradients play a central role in many aspects of machine learning, so too do automatic differentiation systems such as TensorFlow~\cite{abadi2016tensorflow}, PyTorch~\cite{paszke2017automatic}, or Stan~\cite{carpenter2015stan}.

\begin{figure}[b]
    \[
	\xymatrix@C+2mm{
          *+[F]{\txt{Programs}}
          \ar[d]_-{\txt{\footnotesize denotational\\\footnotesize semantics}} \ar[rr]^{\txt{\footnotesize automatic\\\footnotesize differentiation}}
		 && *+[F]{\txt{Programs}} \ar[d]^-{\txt{\footnotesize denotational\\\footnotesize semantics}}\\
		 *+[F]{\txt{Differential\\geometry}}\ar[rr]^{\txt{\footnotesize math\\\footnotesize differentiation}}
		 && *+[F]{\txt{Differential\\geometry}}
	}
  \]
  \caption{Overview of semantics/correctness of AD.\label{fig:intro}}
\end{figure}
Differentiation has a well-developed mathematical theory in terms of differential geometry.
The aim of this paper is
to formalize this connection between differential geometry
and the syntactic operations of AD, particularly for AD methods that calculate higher-order derivatives. 
In this way, we achieve two things: (1)~a compositional, denotational understanding of differentiable programming and AD; (2)~an explanation of the correctness of AD.

This intuitive correspondence (summarized in Fig.~\ref{fig:intro}) is in fact rather complicated.
In this paper, we focus on resolving the following problem: higher-order functions play a key role in programming, and yet they have no counterpart in traditional differential geometry. Moreover, we resolve this problem while retaining the compositionality of denotational semantics.
  
\subsubsection{Higher-order functions and differentiation.}
A major application of higher-order functions is to support disciplined code reuse.
The need for code reuse is particularly acute in machine learning.
For example,
a multi-layer neural network might be built from millions of near-identical neurons,
as follows. 
\newcommand{\neuron}{\mathrm{neuron}}
\newcommand{\layer}{\mathrm{layer}}
\newcommand{\compose}{\mathrm{comp}}
\[\begin{array}{ll}\begin{aligned}
  &\neuron_n:\bProd{\reals^n}{\bProd{\reals^n}{\reals}}\To\reals
  \\&\neuron_n\defeq \lambda \tTuple{x,\tTuple{w,b}}.\,\sigmoid( w\cdot x+b) 
  \\
  &\layer_n:(\bProd{\ty_1}{P}\To \ty_2)\To \bProd{\ty_1}{P^n}\To \ty_2^n
  \\
  &\layer_n\defeq\lambda f.\,\lambda\tTuple{x,\tTuple{p_1,\dots, p_n}}.\,
  \tTuple{f\tTuple{x,p_1},\dots,f\tTuple{x,p_n}}
  \\
  &\compose : \bProd{(\bProd{\ty_1}{P}\To \ty_2)}{(\bProd{\ty_2}{Q}\To \ty_3)}\To \bProd{\ty_1}{\bProd{P}{Q}}\To \ty_3
    \\
  &\compose\defeq\lambda \tTuple{f,g}.\,\lambda\tTuple{x,(p,q)}.\,g\tTuple{f\tTuple{x,p},q}
\end{aligned}
    &\hspace{-16mm}
\raisebox{-8mm}[0pt]
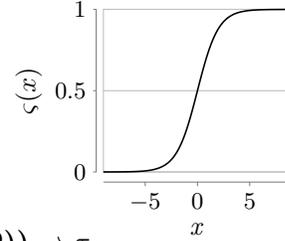
{\begin{tikzpicture}[xscale=0.5,yscale=0.7]
      \datavisualization [scientific axes=clean,
                    y axis=grid,visualize as smooth line,
                    y axis={label={$\sigmoid(x)$},ticks={step=0.5}},
                    x axis={label, ticks={
                                        step = 5}} ]
data [format=function] {
      var x : interval [-9:9] samples 100;
      func y = 1/(1 + exp(-\value x));
      };
      \end{tikzpicture}}
\end{array}\]
(Here $\sigmoid(x) \defeq\frac 1 {1+e^{-x}}$ is the sigmoid function, as illustrated.)
We can use these functions to build a network as follows (see also Fig.~\ref{fig:network}):
\begin{equation}\label{eqn:network}
  \compose\tTuple{\layer_m(\neuron_k),\compose\tTuple{\layer_n(\neuron_m),\neuron_n}}
  :\bProd{\reals^k}{P}\to\reals
\end{equation}
\begin{wrapfigure}[12]{r}{0.35\linewidth} 
\qquad 
\tikzset{every neuron/.style={
    circle,
    draw,
    minimum size=4mm
  },
  neuron missing/.style={
    draw=none, 
    scale=1,
    text height=0.01cm,
    execute at begin node=\color{black}$\cdots$
  },
}

\begin{tikzpicture}[x=1.5cm, y=1.5cm, >=stealth, rotate=90,xscale=-0.3,yscale=0.5]

\foreach \m/\l [count=\y] in {1,2,3,missing,4}
  \node [every neuron/.try, neuron \m/.try] (input-\m) at (0,2.5-\y) {};

\foreach \m [count=\y] in {1,2,missing,3}
  \node [every neuron/.try, neuron \m/.try ] (hidden-\m) at (2,2-\y*1.25) {};

\foreach \m [count=\y] in {1,2,missing,3}
  \node [every neuron/.try, neuron \m/.try ] (hiddenb-\m) at (4,2-\y*1.25) {};

\foreach \m [count=\y] in {1}
  \node [every neuron/.try, neuron \m/.try ] (output-\m) at (6,1.5-\y) {};

\foreach \l [count=\i] in {1,2,3,k}
  \draw [<-] (input-\i) -- ++(-1,0)
    node at (input-\i) {$\l$};

\foreach \l [count=\i] in {1,2,m}
  \node at (hidden-\i) {$\l$};

\foreach \l [count=\i] in {1,2,n}
  \node at (hiddenb-\i) {$\l$};

\foreach \l [count=\i] in {1}
  \draw [->] (output-\i) -- ++(1,0)
    node [above, midway] {};

\foreach \i in {1,...,4}
  \foreach \j in {1,...,3}
    \draw [->] (input-\i) -- (hidden-\j);

\foreach \i in {1,...,3}
  \foreach \j in {1,...,3}
    \draw [->] (hidden-\i) -- (hiddenb-\j);

\foreach \i in {1,...,3}
  \foreach \j in {1}
    \draw [->] (hiddenb-\i) -- (output-\j);

\end{tikzpicture}

 \caption{The network in~\eqref{eqn:network} with $k$ inputs and two hidden layers.\label{fig:network}}\end{wrapfigure}
Here $P\cong \reals^p$ with $p=(m(k{+}1){+}n(m{+}1){+}n{+}1)$. 
This program~\eqref{eqn:network} describes a smooth (infinitely differentiable) function.
The goal of automatic differentiation is to find its derivative.

If we $\beta$-reduce all the $\lambda$'s, we obtain a very long function expression built only from the sigmoid function and linear algebra. We can then find a program for calculating its derivative by applying the chain rule.
However, automatic differentiation can also be expressed without first $\beta$-reducing,
in a compositional way, by explaining how higher-order functions, such as $(\layer)$ and $(\compose)$, propagate derivatives. This paper is a semantic analysis of this compositional approach.

The general idea of denotational semantics is to interpret types as spaces and programs as functions between spaces. In this paper, we propose to use
diffeological spaces and smooth functions~\cite{souriau1980groupes,iglesias2013diffeology} to this end.
These satisfy the following three desiderata:
\begin{itemize}
\item $\RR$ is a space, and the smooth functions $\RR\to\RR$ are exactly the infinitely differentiable functions;
\item The set of smooth functions $X\to Y$ between spaces again forms a space,
  so we can interpret function types.
\item The disjoint union of a sequence of spaces again forms a space, and this enables us to interpret variant types and inductive types, 
e.g. lists of reals form the space $\biguplus_{i=0}^{\infty}\RR^i$.
\end{itemize}
We emphasise that the most standard formulation of differential geometry, using manifolds, does not support spaces of functions. Diffeological spaces seem to us to be the simplest notion of space that satisfies these conditions, but there are other candidates~\cite{baez2011convenient,smootheology}.
A diffeological space is, in particular, a set $X$ equipped with a chosen set of curves
$C_X\subseteq X^\RR$;
a smooth map $f:X\to Y$ must be such that if $\gamma\in C_X$ then
$\gamma;f\in C_Y$. 
This is reminiscent of the method of logical relations.

\subsubsection{From smoothness to automatic derivatives at higher types.}
Our denotational semantics in diffeological spaces
guarantees that all definable functions are smooth.
However, we need more than just to know that a definable function happens to have a mathematical derivative: we need to be able to find that derivative.

In this paper, we focus on forward-mode automatic differentiation methods for computing higher derivatives, which are macro translations of syntax (called~$\Dsynsymbol$ in Section~\ref{sec:simple-language}). 
We are able to show that they are correct, using our denotational semantics.

Here there is one subtle point that is central to our development.
Although differential geometry provides established derivatives for first-order functions
(such as $\neuron$ above),
there is no canonical notion of derivative for higher-order functions (such as $\layer$ and $\compose$) 
in the theory of diffeological spaces (e.g.~\cite{christensen2014tangent}). 
We propose a new way to resolve this by interpreting types as triples $(X,X',S)$ where, intuitively, $X$ is a space of inhabitants of the type, $X'$ is a space serving as a chosen bundle of tangents (or jets, in the case of higher-order derivatives) over $X$, and $S\subseteq X^\RR\times X'^\RR$ is a binary relation between curves, informally relating curves in $X$ with their tangent (resp. jet) curves in $X'$.
This new model gives a denotational semantics for higher-order automatic differentiation on a language with higher-order functions. 

In Section~\ref{sec:semantics} we boil this new approach down to a straightforward and elementary logical relations argument for the correctness of higher-order automatic differentiation. The approach is explained in detail in Section~\ref{sec:correctness}. 
We explore some subtleties of non-uniqueness of derivatives of higher-order functions in Section~\ref{derivatives-at-higher-types}.

\subsubsection{Related work and context.}
AD has a long history and has many implementations.
AD was perhaps first formulated in a functional setting in~\cite{pearlmutter2008reverse}, and there are now a number of teams working on AD in the functional setting 
(e.g.~\cite{wang2018demystifying,shaikhha2019efficient,elliott2018simple}), some providing efficient implementations. 
Although that work does not involve formal semantics, it is inspired by intuitions from differential geometry and category theory. 

This paper adds to a very recent body of work on verified automatic differentiation.
In the first-order setting,
there are recent accounts based on denotational semantics in manifolds~\cite{fong2019backprop, lee2020correctness} and on synthetic differential geometry~\cite{gallagher-sdg},
work on categorical abstractions~\cite{rev-deriv-cat2020}, and
work connecting operational semantics 
with denotational semantics~\cite{abadi-plotkin2020,plotkin-invited-talk},
as well as work focusing on how to correctly differentiate programs that operate on tensors~\cite{bernstein2020differentiating} and programs that use quantum computing~\cite{zhu2020principles}.
Recently, there has also been significant progress at higher types. Brunel et al.~\cite{brunel2019backpropagation} and 
Mazza and Pagani~\cite{mazza2020automatic} give formal correctness proofs for reverse-mode derivatives on a linear $\lambda$-calculus with a particular operational semantics.
The work of Barthe et al.~\cite{bcdg-open-logical-relations} provides a general discussion of some new syntactic logical relations arguments, including one very similar to our syntactic proof of Theorem~\ref{thm:fwd-cor-basic}. 
Sherman et al.~\cite{sherman2020lambda_s} discuss a differential programming technique that works 
at higher types, based on exact real arithmetic, and relate it to a computable semantics.
We understand that the authors of~\cite{gallagher-sdg} are working on higher types.
V\'ak\'ar~\cite{vakar2021reverse,vakar2021chad,lucatelli2021chad} formulates and proves correct a reverse-mode AD technique on a higher-order language based on a similar gluing technique. 
V\'ak\'ar~\cite{vakar2020denotational} extends a standard $\lambda$-calculus with type recursion, and proves correct a forward-mode AD on such a higher-order language, also using a gluing argument.

The differential $\lambda$-calculus~\cite{ehrhard2003differential} is related to AD, and explicit connections are made in~\cite{mak-ong2020,Manzyuk2012}. 
One difference is that the differential $\lambda$-calculus allows the addition of terms at all types, 
and hence vector space models are suitable to interpret all types.
This choice would appear unusual for the variant and inductive types that we consider here, 
as the dimension of a disjoint union of spaces is only defined locally.

This paper builds on our previous work~\cite{huot2020correctness,vakar2020denotational}
in which we gave denotational correctness proofs for forward-mode AD algorithms for computing first derivatives.
Here, we explain how these techniques extend to methods that calculate higher derivatives. 

The Faà di Bruno construction has also been investigated in the context of Cartesian differential categories~\cite{cockett2011faa}.

The idea of directly computing higher-order derivatives using automatic differentiation methods that work with Taylor approximations (also known as jets in differential geometry) is well-known~\cite{griewank2000evaluating}, and it has recently gained renewed interest~\cite{betancourt2018geometric,bettencourt2019taylor}.
So far, such ``Taylor-mode AD'' methods have only been applied to first-order functional languages.
This paper shows how to extend these higher-order AD methods to languages with support for higher-order functions 
and algebraic data types.

The two main methods for implementing AD are operator overloading and source-code transformation; we use the latter in this paper~\cite{van2018automatic}. Taylor-mode AD has been shown to be significantly faster than iterated AD in the context of operator overloading in JAX~\cite{bettencourt2019taylor,frostig2018compiling}. There are other notable implementations of forward Taylor-mode AD~\cite{bendtsen1996fadbad,bendtsen1997tadiff,karczmarczuk2001functional,pearlmutter2007lazy,wang2016capitalizing}. Some of them are implemented in functional languages~\cite{karczmarczuk2001functional,pearlmutter2007lazy}. Taylor-mode implementations use the rich algebraic structure of derivatives to avoid many redundant computations in iterated first-order methods by sharing intermediate results. Perhaps the simplest example is the sine function, whose iterated derivatives only involve sin, cos, and negation. Importantly, most AD tools have the right complexity up to a constant factor, but this constant is quite important in practice and Taylor-mode helps achieve better performance. Another striking result of a version of Taylor-mode was achieved in~\cite{laue2018computing}, where a performance gain of up to two orders of magnitude was achieved for computing certain Hessian-vector products using Ricci calculus. In essence, the algorithm used is a mixed-mode algorithm derived via jets in~\cite{betancourt2018geometric}. This is further improved in~\cite{laue2020simple}.
Taylor-mode can also be useful for ODE solvers and hence will be important for neural differential equations~\cite{chen2018neural}. 

Finally, we emphasise that we have chosen the neural network~(\ref{eqn:network})
as our running example mainly for its simplicity. Indeed, one would typically use reverse-mode AD to train neural networks in practice.
There are many other examples of AD outside the neural networks literature:
AD is useful whenever derivatives need to be calculated in high-dimensional spaces. This includes optimization problems more generally, where the derivative is passed to a
gradient descent method (e.g.~\cite{robbins1951stochastic,kiefer1952stochastic,qian1999momentum,kingma2014adam,duchi2011adaptive,liu1989limited}).
Optimization problems involving higher-order functions naturally show up in the calculus of variations and its applications in physics, where one typically looks for a function minimizing a certain integral~\cite{gelfand2000calculus}.
Other applications of AD are in advanced \emph{integration} methods, since derivatives play a role in 
Hamiltonian Monte Carlo~\cite{neal2011mcmc,hoffman2014no} and variational inference~\cite{kucukelbir2017automatic}. 
Second-order methods for gradient descent have also been extensively studied. As the basic second-order Newton method requires inverting a high-dimensional Hessian matrix, several alternatives and approximations have been studied. 
Some of them still require Taylor-like modes of differentiation and require a matrix-vector product where the matrix resembles the Hessian or inverse Hessian~\cite{knoll2004jacobian, martens2010deep, amari2012differential}. 

\subsubsection{Summary of contributions.}
We have provided a semantic analysis of higher-order automatic differentiation. 
Our syntactic starting point is higher-order forward-mode AD macros on a typed higher-order language 
that extend their well-known first-order counterparts (e.g.~\cite{shaikhha2019efficient,wang2018demystifying,huot2020correctness}). We present these in Section~\ref{sec:simple-language} 
for function types, and in Section~\ref{sec:extended-language} we extend them to inductive types and variants. 
The main contributions of this paper are as follows. 
\begin{itemize}
\item We give a denotational semantics for the language in diffeological spaces, showing that every definable expression is smooth (Section~\ref{sec:semantics}).
\item We show correctness of the higher-order AD macros by a logical relations argument (Theorem~\ref{thm:fwd-cor-basic}).
\item We give a categorical analysis of this correctness argument with two parts: a universal property satisfied by the macro in terms of syntactic categories, and a new notion of glued space that abstracts the logical relation (Section~\ref{sec:correctness}).
\item We then use this analysis to state and prove a correctness argument at all first-order types (Theorem~\ref{thm:fwd-cor-full}). 
\end{itemize}

\subsubsection*{Relation to previous work}
This paper extends and develops the paper~\cite{huot2020correctness} presented at 
the 23rd International Conference on Foundations of Software Science and Computation Structures (FoSSaCS 2020).
This version includes numerous elaborations, 
notably extensions of the definitions, semantics, and correctness proofs for automatic differentiation methods for computing higher-order derivatives (introduced in Sections~\ref{sub:hod}--\ref{sub:hod-ex}) 
and a novel discussion of derivatives of higher-order functions (Section~\ref{derivatives-at-higher-types}).

\section{Rudiments of differentiation: how to calculate with dual numbers and Taylor approximations}\label{sec:dual-numbers-taylor}

\subsection{First-order differentiation: the chain rule and dual numbers.}
We now recall the definition of the gradient of a differentiable function, the goal of AD, and what it means for AD to be correct.
Recall that the derivative of a function $f:\RR\to \RR$, if it exists, is a function
$\nabla f:\RR\to \RR$ such that for all $a$, $\nabla f(a)$ is the gradient of $f$ at $a$ in the sense that the function
$x\mapsto f(a) +\nabla f(a)\cdot (x-a)$ gives the best linear approximation of $f$ at $a$.
(The gradient $\nabla f(a)$ is often written $\frac {\dif f(x)}{\dif x}(a)$.)

The chain rule for differentiation tells us that we can calculate $\nabla (f;g)(a) = \nabla f(a)\cdot \nabla g(f(a))$.
In that sense, the chain rule tells us how linear approximations to a function transform under post-composition with another function. 

To find $\nabla f$ in a compositional way, using the chain rule, two generalizations are useful:
\begin{itemize}
\item We need both $f$ and $\nabla f$ when calculating $\nabla (f;g)$
of a composition $f;g$, using the chain rule, so we are really interested in the pair $(f,\nabla f):\RR\to \RR\times \RR$;
\item In building $f$ we will need to consider functions of multiple arguments, such as $+:\RR^2\to \RR$, and these functions should propagate derivatives.
\end{itemize}
Thus we are more generally interested in transforming a function $g:\RR^n\to \RR$ into a function
$h:(\RR\times \RR)^n\to \RR\times \RR$ in such a way that for any
$f_1\dots f_n:\RR\to\RR$, 
\begin{equation}
  \label{eqn:dualnumber}
  (f_1,\nabla f_1,\dots, f_n,\nabla f_n);h
  =
  ((f_1,\dots, f_n);g,\nabla ((f_1, \dots, f_n);g))\text.
\end{equation}

Automatically computing a program representing $h$, given a program representing $g$, is the goal of automatic differentiation.
An intuition for $h$ is often given in terms of dual numbers.
The transformed function operates on pairs of numbers, $(x,x')$, and it is common
to think of such a pair as $x+x'\epsilon$ for an `infinitesimal' $\epsilon$.
Although this is a helpful intuition, the formalization of infinitesimals can be intricate, 
and the development in this paper is focused on the elementary formulation in~\eqref{eqn:dualnumber}.

A function $h$ satisfying (\ref{eqn:dualnumber}) encodes all the partial derivatives of
$g$. For example, 
if $g \colon \RR^2\to \RR$, then with $f_1(x)\defeq x$ and $f_2(x)\defeq x_2$, by applying \eqref{eqn:dualnumber} to $x_1$ we obtain
$h(x_1,1,x_2,0)\!=\!(g(x_1,x_2), \frac {\partial g(x,x_2)}{\partial x}(x_1))$
and similarly 
$h(x_1,0,x_2,1)\!=\!(g(x_1,x_2), \frac {\partial g(x_1,x)}{\partial x}(x_2))$.
Conversely, if $g$ is differentiable as a function $\RR^2\to\RR$, then
a unique $h$ satisfying \eqref{eqn:dualnumber} can be found by taking linear
combinations of partial derivatives, for example:
\[\textstyle h(x_1,x_1',x_2,x_2')=(g(x_1,x_2),x_1' \cdot\frac {\partial g(x,x_2)}{\partial x}(x_1)+x_2'\cdot \frac {\partial g(x_1,x)}{\partial x}(x_2))\text.\]
(Here, recall that the partial derivative  $\frac{\partial g(x,x_2)}{\partial x}(x_1)$
is a particular notation for the gradient $\nabla(g(-,x_2))(x_1)$, i.e.~with $x_2$ fixed. 
)

In summary, the idea of differentiation with dual numbers is 
to transform a differentiable function
$g:\RR^n\to \RR$ to a function $h:\RR^{2n}\to \RR^2$ that captures~$g$ and all its partial derivatives. We express this in~\eqref{eqn:dualnumber} as an invariant that is useful for building derivatives of compound functions $\RR\to\RR$ in a compositional way.
The idea of (first-order) forward-mode automatic differentiation is to perform this transformation at the source-code level. 

We say that a macro for AD is correct if, given a semantic model $\sem{-}$, 
the program $P$ representing $g=\sem{P}$ is transformed by the macro to a program $P'$ representing $h=\sem{P'}$.
This means in particular that $P'$ computes correct partial derivatives of the function represented by $P$.

\subsubsection*{Smooth functions.}
In what follows, we will often speak of \emph{smooth} functions $\RR^k\to\RR$, namely functions that are continuous and differentiable, with derivatives that are also continuous and differentiable, and so on. 

\subsection{Higher-order differentiation: the Fa\`a di Bruno formula and Taylor approximations.}
\label{sub:hod}
We now generalize the above in two directions:
\begin{itemize}
\item We look for the best local approximations to $f$ with polynomials of some order~$R$, generalizing the above use of linear functions ($R=1$).
\item We can work directly with multivariate functions $\RR^k\to \RR$ instead of functions of one variable $\RR\to\RR$ ($k=1$).
\end{itemize}
To make this precise, we recall that, given a smooth function 
$f:\RR^k\to \RR$ and a natural number $R\geq 0$, the \emph{$R$-th order Taylor approximation of $f$  at ${a\in \RR^k}$} is defined in terms of the partial derivatives 
of $f$:
\begin{align*}
\RR^k\;\;\quad&\to\qquad\RR\\
{x}\qquad&\mapsto\hspace{-30pt} \sum_{\set{(\alpha_1,\ldots,\alpha_k)\in \NN^k\mid \alpha_1+\ldots+\alpha_k\leq R}}{\frac{1}{\alpha_1!\cdot \ldots\cdot \alpha_k!}
\frac{\partial^{\alpha_1+\ldots+\alpha_k}f(x)}{\partial x_1^{\alpha_1}\cdots\partial x_k^{\alpha_k}}(a)}\cdot
(x_1-a_1)^{\alpha_1}\cdot\ldots\cdot (x_k-a_k)^{\alpha_k}.
\end{align*}
This is an $R$-th order polynomial. 
Similarly to the case of first-order derivatives, the partial derivatives of $f$ up to order~$R$ can be read off from the coefficients of its Taylor approximation, up to the displayed factorial factors. See Section~\ref{sub:two-dim-taylor} below for an example.

Recall that the ordering of partial derivatives does not matter for smooth functions (Schwarz/Clairaut's theorem). 
So there will be $\binom {R+k-1} {k-1}$ $R$-th order partial derivatives, and altogether there are $\binom {R+k} k$ summands in the $R$-th order Taylor approximation.
(This can be seen by a `stars-and-bars' argument.)

Since there are ${\cchoose{R+k}k}$ partial derivatives of $f$ of order $\leq R$, we can store them 
in the Euclidean space $\RR^{\cchoose{R+k}k}$, which can also be regarded as the space of $k$-variate polynomials of degree~$\leq R$.

We use a convention of coordinates $\seq[(\alpha_1,\ldots,\alpha_k)\in\set{(\alpha_1,\ldots,\alpha_k)\in \NN^k\mid 0\leq \alpha_1+\ldots+\alpha_k\leq R}]{y_{\alpha_1...\alpha_k}\in \RR}$ where
$y_{\alpha_1\ldots \alpha_k}$ is intended to represent a partial derivative $\frac{\partial^{\alpha_1+...+\alpha_k}f}{\partial x_1^{\alpha_1}\cdots \partial x_k^{\alpha_k}}(a)$ for some function $f:\RR^k\to \RR$.
We will choose these coordinates in lexicographic order of the multi-indices $(\alpha_1,\ldots,\alpha_k)$, that is, the indices in the Euclidean space $\RR^{\cchoose {R+k}k}$ will typically range from $(0,\ldots,0)$ to $(R,0,\ldots,0)$.

The \emph{$(k,R)$-Taylor representation} of a function $g:\RR^n\to \RR$ is a function
$h: \left(\RR^{\cchoose{R+k}k}\right)^n\to \RR^{\cchoose{R+k}k}$ that transforms the 
partial derivatives of $f:\RR^k \to \RR^n$ of order $\leq R$ under postcomposition with $g$:
\begin{equation}
  \label{eqn:taylorrepresentation}
  \resizebox{\linewidth}{!}{\parbox{1.1\linewidth}{\[
  {\left({\left(
    \frac{\partial^{\alpha_1+\ldots+\alpha_k}f_j(x)}{\partial x_1^{\alpha_1}\cdots \partial x_k^{\alpha_k}}\right)}_{(\alpha_1,...,\alpha_k)=(0,...,0)}^{(R,0,...,0)}\right)}_{j=1}^n;h
  =
  {\left(
    \frac{\partial^{\alpha_1+\ldots+\alpha_k}((f_1,\ldots,f_n);g)(x)}{\partial x_1^{\alpha_1}\cdots \partial x_k^{\alpha_k}}\right)}_{(\alpha_1,...,\alpha_k)=(0,...,0)}^{(R,0,...,0)}\hspace{-12pt}
  \text.\]}}
\end{equation}
Thus the Taylor representation generalizes the dual numbers representation ($R=k=1$).

To calculate the Taylor representation for a smooth function explicitly, we recall a generalization of the chain rule to higher derivatives. The chain rule tells us how the coefficients of linear approximations transform under composition of functions. 
The \emph{Fa\`a di Bruno formula}~\cite{savits2006some,encinas2003short,constantine1996multivariate}
tells us how the coefficients of Taylor approximations
-- that is, higher derivatives -- transform 
under composition.
We recall the multivariate form from~\cite[Theorem~2.1]{savits2006some}. 
Given functions $f=(f_1,\ldots,f_l):\RR^k\to \RR^l$ and $g:\RR^l\to\RR$, for 
$\alpha_1+\ldots+\alpha_k>0$, write $\alpha=(\alpha_1,\ldots,\alpha_k)$ and write $|\gamma|$ for the sum of the components of a multi-index $\gamma$.

\begin{align*}
    \frac{\partial^{|\alpha|}(f;g)(x)}
         {\partial x_1^{\alpha_1}\cdots\partial x_k^{\alpha_k}}(a)
    &= \Bigl(\prod_{i=1}^k \alpha_i!\Bigr)
       \sum_{\substack{\beta\in \NN^l\\ 1\leq |\beta|\leq |\alpha|}}
       \frac{\partial^{|\beta|}g(y)}
            {\partial y_1^{\beta_1}\cdots \partial y_l^{\beta_l}}(f(a)) \\
    &\quad{}\cdot
       \sum_{\substack{e^1,\ldots,e^q\in\NN^l\\
          e^1+\cdots+e^q=\beta\\
          \sum_{r=1}^q |e^r|\alpha_i^r=\alpha_i\ (1\leq i\leq k)}}
       \prod_{r=1}^q
       \prod_{j=1}^l \frac{1}{e^r_j!}
       \Biggl(
          \frac{1}{\alpha^r_1!\cdots \alpha^r_k!} \cdot
          \frac{\partial^{|\alpha^r|} f_j(x)}
               {\partial x_1^{\alpha^r_1}\cdots\partial x_k^{\alpha^r_k}}(a)
       \Biggr)^{e^r_j},
\end{align*}
where $(\alpha^1_1,\ldots,\alpha^1_k),\ldots, (\alpha^q_1,\ldots,\alpha^q_k)\in \NN^k$ are an enumeration of all the vectors $(\alpha^r_1,\ldots,\alpha^r_k)$ of $k$ natural numbers such that $\alpha^r_i\leq \alpha_i$ for all $i$ and $\alpha^r_1+\ldots+\alpha^r_k>0$ and we write $q$ for the number of such vectors.
The details of this formula reflect the complicated combinatorics arising from the repeated applications of the chain and product rules used to prove it.
Conceptually, however, it is rather straightforward: it tells us that the coefficients of the $R$-th order Taylor approximation of $f;g$ can be expressed exclusively in terms of those of $f$ and $g$.

Thus the Fa\`a di Bruno formula uniquely determines the Taylor approximation $h: \left(\RR^{\cchoose{R+k}k}\right)^n\to \RR^{\cchoose{R+k}k}$ in terms of the derivatives of $g:\RR^n\to \RR$ of order $\leq R$, and we can also recover all such derivatives from $h$.

\subsection{Example: a two-dimensional second-order Taylor series}
\label{sub:two-dim-taylor}
As an example, we can specialize the Fa\`a di Bruno formula above to the second-order Taylor series of a function $f:\RR^2\to \RR^l$ and its behaviour under postcomposition with a smooth function $g:\RR^l\to \RR$:
\begin{align*}
  \frac{\partial^{2}(f;g)(x)}{\partial x_{i}\partial x_{i'}}(a)
  &= \sum_{j=1}^l \frac{\partial g(y)}{\partial y_j}(f(a)) \frac{\partial^2 f_j(x)}{\partial x_{i}\partial x_{i'}}(a)+
  \sum_{j,j'=1}^l \frac{\partial^2 g(y)}{\partial y_j\partial y_{j'}}(f(a))\frac{\partial f_{j'}(x)}{\partial x_{i}}(a) 
  \frac{\partial f_j(x)}{\partial x_{i'}}(a),
\end{align*}
where $i,i'\in \set{1,2}$ may coincide or be distinct.

Rather than working with the full $(2,2)$-Taylor representation of $g$, we ignore the non-mixed second-order derivatives  
$y_{02}^j=\frac{\partial^2 f_j(x)}{\partial x_2^2}$ and  $y_{20}^j=\frac{\partial^2 f_j(x)}{\partial x_1^2}$ for the moment, and we represent the derivatives of order $\leq 2$ of $f_j: \RR^2\to \RR$ (at some point $a$) as the numbers 
$$(y_{00}^j, y_{01}^j, y_{10}^j,y_{11}^j)=\left(f_j(a),\frac{\partial f_j(x)}{\partial x_2}(a),\frac{\partial f_j(x)}{\partial x_1}(a),\frac{\partial^2 f_j(x)}{\partial x_1 \partial x_2}(a)   \right)\in \RR^4$$
and we can choose a similar representation for the derivatives of $(f;g)$.
We observe that the Fa\`a di Bruno formula induces the function $h:(\RR^4)^l\to \RR^4$
\begin{align*}
 & h((y_{00}^1,y_{01}^1,y_{10}^1,y_{11}^1),\ldots, (y_{00}^l,y_{01}^l,y_{10}^l,y_{11}^l))=\\
&\left(
\begin{array}{l}
g(y_{00}^1,\ldots,y_{00}^l)\\
\sum_{j=1}^l \frac{\partial g(y^1,\ldots,y^l)}{\partial y_j}(y_{00}^1,\ldots,y_{00}^l)\cdot y_{01}^j\\
\sum_{j=1}^l \frac{\partial g(y^1,\ldots,y^l)}{\partial y_j}(y_{00}^1,\ldots,y_{00}^l)\cdot y_{10}^j\\
\sum_{j=1}^l \frac{\partial g(y^1,\ldots,y^l)}{\partial y_j}(y_{00}^1,\ldots,y_{00}^l)\cdot y_{11}^j + 
\sum_{j,j'=1}^l \frac{\partial^2 g(y^1,\ldots,y^l)}{\partial y_j \partial y_{j'}}(y_{00}^1,\ldots,y_{00}^l) \cdot y_{10}^j\cdot y_{01}^{j'}
\end{array}\right).
\end{align*}
In particular, we note that
\begin{align*}
  & h((y_{00}^1,y_{01}^1,y_{10}^1,0),\ldots, (y_{00}^l,y_{01}^l,y_{10}^l,0))=\left(
 \begin{array}{l}
 g(y_{00}^1,\ldots,y_{00}^l)\\
 \sum_{j=1}^l \frac{\partial g(y^1,\ldots,y^l)}{\partial y_j}(y_{00}^1,\ldots,y_{00}^l)\cdot y_{01}^j\\
 \sum_{j=1}^l \frac{\partial g(y^1,\ldots,y^l)}{\partial y_j}(y_{00}^1,\ldots,y_{00}^l)\cdot y_{10}^j\\
 \sum_{j,j'=1}^l \frac{\partial^2 g(y^1,\ldots,y^l)}{\partial y_j \partial y_{j'}}(y_{00}^1,\ldots,y_{00}^l) \cdot y_{10}^j\cdot y_{01}^{j'}
 \end{array}\right).
 \end{align*}
We see that we can use this method to 
calculate any directional first- and second-order derivatives of $g$ in one pass.
For example, if $l=3$, so $g:\RR^3\to \RR$, then the last component of 
$h((x,x',x'',0),(y,y',y'',0),(z,z',z'',0))$ is the result of
taking the first derivative in direction
$(x',y',z')$ and the second derivative in direction $(x'',y'',z'')$, and evaluating at $(x,y,z)$.

In the proper Taylor representation we explicitly include the non-mixed second-order derivatives as inputs and outputs, leading to a function $h':(\RR^6)^l\to \RR^6$. Above, we have followed a common trick to avoid some unnecessary storage and computation, since these extra inputs and outputs are not required for computing the second-order derivatives of~$g$. For instance, if $l=2$ then the last component of $h((x,1,1,0),(y,0,0,0))$ computes $
\frac{\partial^2g}{\partial x^2}(x,y)$.

\subsection{Example: a one-dimensional second-order Taylor series}
\label{sub:hod-ex}
As opposed to (2,2)-AD, (1,2)-AD computes the first- and second-order derivatives in the same direction. For example, if $g:\RR^2\to\RR$ is a smooth function, then $h:(\RR^3)^2\to \RR^3$. An intuition for $h$ can be given in terms of triple numbers.
The transformed function operates on triples of numbers, $(x,x',x'')$, and it is common
to think of such a triple as $x+x'\epsilon+\frac{1}{2}x''\epsilon^2$ for an `infinitesimal' $\epsilon$ with the property that $\epsilon^3=0$. For instance, we have
\begin{align*}
&h((x_1,1,0),(x_2,0,0))=(g(x_1,x_2),\frac {\partial g(x,x_2)}{\partial x}(x_1),
\frac {\partial^2 g(x,x_2)}{\partial x^2}(x_1))\\ 
&h((x_1,0,0),(x_2,1,0))=(g(x_1,x_2),\frac {\partial g(x_1,x)}{\partial x}(x_2),
\frac {\partial^2 g(x_1,x)}{\partial x^2}(x_2))\\
&h((x_1,1,0),(x_2,1,0))=(g(x_1,x_2),\frac {\partial g(x,x_2)}{\partial x}(x_1)+\frac {\partial g(x_1,x)}{\partial x}(x_2),\\
&\qquad\qquad\qquad\qquad\qquad\qquad\qquad\quad
\frac {\partial^2 g(x,x_2)}{\partial x^2}(x_1)+
\frac {\partial^2 g(x_1,x)}{\partial x^2}(x_2)+
2\frac {\partial^2 g(x,y)}{\partial x\partial y}(x_1,x_2))
\end{align*}
We directly get non-mixed second-order partial derivatives but not the mixed ones. We can recover $\frac {\partial^2 g(x,y)}{\partial x \partial y}(x_1,x_2)$ as the third component of $\frac{1}{2}(h((x_1,1,0),(x_2,1,0))-h((x_1,1,0),(x_2,0,0))-h((x_1,0,0),(x_2,1,0)))$.

More generally, if $g:\RR^l\to\RR$, then $h:(\RR^3)^l\to \RR^3$ satisfies:
\begin{align*}
& h((x_1,x'_1,0),\ldots, (x_l,x'_l, 0))=\left(
\begin{array}{l}
g(x_1,\ldots,x_l)\\
\sum_{i=1}^l\frac{\partial g}{\partial x_i}(x_1,\ldots,x_l)\cdot x'_i\\
\sum_{i,j=1}^l \frac{\partial^2 g}{\partial x_i \partial x_j}(x_1,\ldots,x_l)\cdot x'_i\cdot x'_j
\end{array}
\right).
\end{align*}
We can always recover the mixed second-order partial derivatives from this, but doing so requires several computations involving $h$. This therefore differs from the (2,2) method, which is more direct.

\subsection{Remark}
In the rest of this article, we study forward-mode $(k,R)$-automatic differentiation for a language with higher-order functions. The reader may like to fix $k=R=1$ for standard automatic differentiation with first-order derivatives, based on dual numbers. This is the approach taken in the conference version of this paper~\cite{hsv-fossacs2020}. The generalization to higher-order derivatives with arbitrary~$k$ and~$R$ flows straightforwardly through the whole narrative.

\section{A Higher-Order Forward-Mode AD Translation}\label{sec:simple-language}

\subsection{A simple language of smooth functions.}

We consider a standard higher-order typed language with a first-order type $\reals$ of real numbers. The types $(\ty,\ty[2])$ and terms $(\trm,\trm[2])$ are as follows.

\begin{figure}[H]
  {\scalebox{1}{\begin{minipage}{\linewidth}\noindent\begin{syntax}
    \ty, \ty[2], \ty[3] & \gdefinedby & & \syncat{types}                          \\
    &\gor& \reals                      & \synname{real numbers}\\
    &\gor&\tProd{\ty_1}{\dots}{\ty_n} & \synname{finite product} \\
  &\gor& \ty \To \ty[2]              & \synname{function}      \\[6pt]
    \trm, \trm[2], \trm[3] & \gdefinedby & & \syncat{terms}                          \\
    &    & \var                          & \synname{variable}                        \\
    &\gor& \op(\trm_1,\ldots,\trm_n)
    & \synname{operations (including constants)}                      \\
    &\gor& \tTriple{\trm_1}{\dots}{\trm_n}\ \gor  \tMatch{\trm}{\var_1,\dots, \var_n}{\trm[2]}\hspace{-10pt} \;& \synname{tuples/pattern matching}\\
    &\gor& \fun \var    \trm
    \ \gor  \trm\, \trm[2]               & \synname{function abstraction/application}\\

\end{syntax} \end{minipage}}}
  \end{figure}

\noindent The typing rules are in Figure~\ref{fig:types1}.
We have included some abstract basic $n$-ary operations $\op\in\Op_n$ for every $n\in\NN$.
These are intended to include the usual (smooth) mathematical operations that are used in programs to which automatic differentiation is applied. For example,
\begin{itemize}
  \item for any real constant $c\in\RR$, we typically include a constant $\cnst\in\Op_0$; we slightly abuse notation and will simply write $\cnst$ for $\cnst()$ in our examples;
  \item we include some unary operations such as  $\sigmoid\in\Op_1$ which we intend to stand for the usual sigmoid function,
  $\sigmoid(x)\defeq\frac 1 {1+e^{-x}}$;
  \item we include some binary operations such as addition and multiplication $(+),(*)\in\Op_2$;
\end{itemize}
We add some simple syntactic sugar $t-u\defeq
t+\underline{(-1)}* u$ and, for some natural number $n$, \[n\cdot \trm \defeq \overbrace{\trm+...+\trm}^{\text{$n$ times}}
  \qquad
  \text{and}
  \qquad \trm^n\defeq \overbrace{\trm * ... * \trm}^{\text{$n$ times}}
  \]
Similarly, we will frequently denote repeated sums and products using $\sum$- and $\prod$-signs, respectively: for example, we write $\trm_1+...+\trm_n$ as $\sum_{i\in \{1,...,n\}}\trm_i$ and $\trm_1*...*\trm_n$ as $\prod_{i\in \set{1,...,n}}\trm_i$.
This is in addition to programming sugar such as $\letin{\var}{\trm}{\trm[2]}$ for $(\fun{\var}{\trm[2]})\,\trm$
and $\fun{\tTuple{\var_1,\ldots,\var_n}}{\trm}$ for $\fun{\var}{\tMatch{\var}{\var_1,\ldots,\var_n}{\trm}}$.

\begin{figure}[b]
  \oframebox{\scalebox{1}{\begin{minipage}{\linewidth}\noindent\[
  \begin{array}{c}
  \inferrule{
    \Ginf {\trm_1}{\reals}\;\;\dots
    \;\; \Ginf {\trm_n}{\reals}
  }{
    \Ginf {\op(\trm_1,\ldots,\trm_n)}\reals
  }(\op\in\Op_n)
 \\[12pt]
 \\
  \inferrule{
    \Ginf {\trm_1}{\ty_1}\;\;\dots
    \;\;
    \Ginf {\trm_n}{\ty_n}
  }{
    \Ginf{\tTriple{\trm_1}{\dots}{\trm_n}}{\tProd{\ty_1}{\dots}{\ty_n} }
  }
  \qquad
    \inferrule{
    \Ginf{\trm}{\tProd{\ty[2]_1}{\dots}{\ty[2]_n} }
    \;\;
  \Ginf[,{\var[1]_1 \colon \ty[2]_1, {.}{.}{.}, \var[1]_n\colon\ty[2]_n}]{\trm[2]}\ty
  }{
    \Ginf{\tMatch
           \trm
           {\var[1]_1,\dots,\var[1]_n}
           {\trm[2]}}{\ty}
  }
 \\[12pt]
 \\
  \inferrule{
    ~
  }{
    \Ginf \var\ty
  }((\var : \ty) \in \ctx)
  \qquad
  \inferrule{
    \Ginf[, \var : \ty]{\trm}{\ty[2]}
  }{
    \Ginf{\fun{\var:\ty}\trm}{\ty\To\ty[2]}
  }\qquad
  \inferrule{
    \Ginf{\trm}{\ty[2]\To\ty}
    \\
    \Ginf{\trm[2]}{\ty[2]}
  }{
    \Ginf{\trm\, \trm[2]}{\ty}
  }
\end{array}
\]
 \end{minipage}}}
  \caption{Typing rules for the simple language.\label{fig:types1}}
  \end{figure}

\subsection{Syntactic automatic differentiation: a functorial macro.}
\label{sec:admacro}
The aim of higher-order forward-mode AD is to find the $(k,R)$-Taylor representation of a function by syntactic manipulations, for some choice of $(k,R)$ that we fix.
For our simple language, we implement this as the following inductively defined macro $\Dsynsymbol[(k,R)]$ on both types and terms
(see also~\cite{wang2018demystifying,shaikhha2019efficient}).
For the sake of legibility, we simply write $\Dsynsymbol[(k,R)]$ as $\Dsynsymbol$ here and leave the 
dimension $k$ and order $R$ of the Taylor representation implicit.
The following definition is for general $k$ and $R$, but we treat specific cases afterwards in Example~\ref{ex:oneonetwotwo}.
\begin{figure}[H]
\noindent\begin{align*}
&\Dsyn{\ty\To\ty[2]} \defeq \Dsyn{\ty}\To\Dsyn{\ty[2]}\qquad 
\Dsyn{\ty_1\t*...\t*\ty_n} \defeq {\Dsyn{\ty_1}}\t*...\t*{\Dsyn{\ty_n}}\\ 
&\Dsyn{\reals} \defeq {\reals}^{\cchoose{R+k}k} \quad
\text{(i.e.,~the type of tuples of reals of length $\textstyle\cchoose{R+k}{k}$)}  
\end{align*} \end{figure}
\begin{figure}[H]
\noindent\begin{align*}
&\Dsyn{\var} \defeq \var\hspace{80pt}\Dsyn{\cnst} \defeq \tTuple{\cnst,\underline{0},\ldots,\underline{0}}\\
&\Dsyn{\fun \var    \trm} \defeq \fun\var{\Dsyn{\trm}}\hspace{12pt}
\Dsyn{\trm\, \trm[2] } \defeq 
\Dsyn{\trm}\,\Dsyn{\trm[2]}\hspace{12pt}
\Dsyn{\tTriple{\trm_1}{\dots}{\trm_n}} \defeq \tTriple{\Dsyn{\trm_1}}{\dots}{\Dsyn{\trm_n}} \\
&\Dsyn{{\tMatch{\trm}{\var_1,\dots,\var_n}{\trm[2]}}} \defeq
\tMatch{\Dsyn\trm}{\var_1,\dots,\var_n }{\Dsyn{\trm[2]}} \\[6pt]
&\hspace{-4pt}\begin{array}{ll}
    \Dsyn{\op(\trm_1,\ldots,\trm_n)}\defeq~ &\tMatch{\Dsyn{\trm_1}}{\var_{0...0}^1,...,\var_{R,0...0}^1}{}\\
    &\vdots\\
                       &{ \tMatch{\Dsyn{\trm_n}}{\var_{0...0}^n,...,\var_{R,0...0}^n}{}}
                       \\ &
\begin{array}{lll}
    \langle&\hspace{-8pt}D^{0...0}\op(\var_{0...0}^1,...,\var_{R,0...0}^1,...,\var_{0...0}^n,...,\var_{R,0...0}^n),\\
&\cdots ,\\ 
&\hspace{-8pt}D^{R...0}\op(\var_{0...0}^1,...,\var_{R,0...0}^1,...,\var_{0...0}^n,...,\var_{R,0...0}^n)&\hspace{-8pt}\rangle\end{array}
\end{array}\end{align*}
\begin{align*}
&\text{where } \\
&D^{0...0}\op(x_{0...0}^1,...,x_{R,0...0}^1,...,x_{0...0}^n,...,x_{R,0...0}^n)\defeq \op(x_{0...0}^1,...,x_{0...0}^n)\\
& \begin{array}{@{}l}
    D^{\alpha_1...\alpha_k}\op(x_{0...0}^1,...,x_{R,0...0}^1,...,x_{0...0}^n,...,x_{R,0...0}^n)\defeq~ \qquad\text{\scriptsize(for $\alpha_1+...+\alpha_k>0$)}\\
                       \begin{array}{c}
                       {\alpha_1! \cdot \ldots \cdot \alpha_k!}\cdot
                       \sum_{\set{(\beta_1,\ldots,\beta_n)\in \NN^n\mid 1\leq \beta_1+\ldots+\beta_n\leq \alpha_1+\ldots+\alpha_k}} 
                       \partial_{\beta_1\cdots\beta_n}\op(x_{0...0}^1,\ldots,x_{0...0}^n)* \\
                       \sum_{\set{((e^1_1,\ldots,e^1_n),\ldots,(e^q_1,\ldots,e^q_n))\in (\NN^n)^q\mid 
                       e_j^1+\ldots+e_j^q = \beta_j\text{ for all }1\leq j\leq n,
                       (e^1_1+\ldots+e^1_n)\cdot \alpha^1_i+ \ldots + (e^q_1+\ldots+e^q_n)\cdot \alpha^q_i=\alpha_i\text{ for all }1\leq i\leq k}}\\
                       \prod_{r=1}^q   
                       \prod_{j=1}^n{\frac{1}{e^r_j!}}\cdot\left({\frac{1}{\alpha^r_1!\cdot\ldots\cdot \alpha^r_k!}}\cdot
                       x_{\alpha^r_1\cdots\alpha^r_k}^j\right)^{e^r_j}\text.                   
                       \end{array}
\end{array}
\end{align*}
 \end{figure}
Here, $(\partial_{\beta_1\cdots \beta_n}\op)(x_1,\ldots,x_n)$ are some chosen terms of type $\reals$ in the language
with free variables among $x_1,\ldots,x_n$.
We think of these terms as implementing the partial derivative $\frac{\partial^{\beta_1+...+\beta_n}\sem{\op}(x_1,...,x_n)}{\partial x_1^{\beta_1}\cdots \partial x_n^{\beta_n}}$ of the smooth function $\sem{\op}:\RR^n\to\RR$ that $\op$ implements.
In the correctness statements below, we assume that the chosen terms have these denotations.
For example, we could choose the following representations of derivatives of order $\leq 2$ of our example operations
\[
    \begin{array}{ll}
\partial_{01}(+)(x_1,x_2) = \underline{1} &
\partial_{02}(+)(x_1,x_2) = \underline{0}\\ 
\partial_{10}(+)(x_1,x_2) = \underline{1}& 
\partial_{11}(+)(x_1,x_2) = \underline{0}\\
\partial_{20}(+)(x_1,x_2) = \underline{0}\\[6pt]
\partial_{01}(*)(x_1,x_2) = x_1 &
\partial_{02}(*)(x_1,x_2) = \underline{0}\\
\partial_{10}(*)(x_1,x_2) = x_2 & 
\partial_{11}(*)(x_1,x_2) = \underline{1}\\
\partial_{20}(*)(x_1,x_2) = \underline{0}\\[6pt]
\partial_{1}(\sigmoid)(\var) = \letin{\var[2]}{\sigmoid(\var)}{
   \var[2]*(\underline{1}-\var[2])}{}& 
\partial_{2}(\sigmoid)(\var) =\letin{\var[2]}{\sigmoid(\var)}{\\
&\qquad \qquad\quad\letin{\var[3]}{\var[2]*(\underline{1}-\var[2])}{\var[3]*(\underline{1}-\underline{2}*\var[2])}}
    \end{array}
\] Note that our rules, in particular, imply that $\Dsyn{\cnst}=\tTuple{\cnst,\underline{0},\ldots,\underline{0}}$.

\begin{example}[$(1,1)$- and $(2,2)$-AD]\label{ex:oneonetwotwo}
  Our choices of partial derivatives of the example operations are sufficient to implement $(k,R)$-Taylor forward AD with $R\leq 2$.
  To be explicit, the distinctive formulas for $(1,1)$- and $(2,2)$-AD methods (specializing our abstract definition of $\Dsynsymbol[(k,R)]$ above) are 
  \begin{align*}
  &\Dsyn[(1,1)]{\reals}=\bProd{\reals}{\reals}\\
  &\Dsyn[(1,1)]{\op(\trm_1,\ldots,\trm_n)}=
    \\&\qquad\begin{array}{l}\pMatch{\Dsyn[(1,1)]{\trm_1}}{\var^1_{0}}{\var^1_1}
                         { \ldots \to\pMatch{\Dsyn[(1,1)]{\trm_n}}{\var_0^n}{\var_1^n}
                         {\\ \tPair{\op(\var_0^1,\ldots,\var_0^n)}{\sum_{i=1}^n\var^i_1 *\partial_i\op(\var_0^1,\ldots,\var_0^n)}}}
  \end{array}
  \\[6pt]
  &\Dsyn[(2,2)]{\reals}=\reals^6\\
  &\Dsyn[(2,2)]{\op(\trm_1,...,\trm_n)}=\\&\qquad\begin{array}{l}
\tMatch{\Dsyn[(2,2)]{\trm_1}}{\var^1_{00},\var^1_{01},\var^1_{02},\var^1_{10},\var^1_{11},\var^1_{20}}{}\\
\vdots\\
\tMatch{\Dsyn[(2,2)]{\trm_n}}{\var^n_{00},\var^n_{01},\var^n_{02},\var^n_{10},\var^n_{11},\var^n_{20}}{}\\\qquad
\begin{array}{l}\langle\op(\var^1_{00},\ldots,\var^n_{00}),\\
\sum_{i=1}^n\var^i_{01} *\partial_{\hat{i}}\op(\var_{00}^1,\ldots,\var_{00}^n),\\
\sum_{i=1}^n\var^i_{02} *\partial_{{\hat{i}}}\op(\var_{00}^1,\ldots,\var_{00}^n)+
\sum_{i,j=1}^n\var^i_{01}*\var^{j}_{01}*\partial_{\widehat{{i,j}}}\op(\var_{00}^1,\ldots,\var_{00}^n),\\
\sum_{i=1}^n\var^i_{10} *\partial_{\hat{i}}\op(\var_{00}^1,\ldots,\var_{00}^n),\\
\sum_{i=1}^n\var^i_{11} *\partial_{{\hat{i}}}\op(\var_{00}^1,\ldots,\var_{00}^n)+
\sum_{i,j=1}^n\var^i_{10}*\var^j_{01}*\partial_{\widehat{{i,j}}}\op(\var_{00}^1,\ldots,\var_{00}^n),\\
\sum_{i=1}^n\var^i_{20} *\partial_{{\hat{i}}}\op(\var_{00}^1,\ldots,\var_{00}^n)+
\sum_{i,j=1}^n\var^i_{10}*\var^j_{10}*\partial_{\widehat{{i,j}}}\op(\var_{00}^1,\ldots,\var_{00}^n)\rangle
\end{array}
\end{array}
\end{align*}
where we informally write $\hat{i}$ for the one-hot encoding of $i$ (the sequence of length $n$ consisting exclusively of zeros except at position $i$, where it has a $1$)
and $\widehat{i,j}$ for the two-hot encoding of $i$ and $j$ (the sequence of length $n$ consisting exclusively of zeros except at positions $i$ and $j$, where it has a $1$ if $i\neq j$ and a $2$ if $i=j$).

As noted in Section~\ref{sec:dual-numbers-taylor}, it is often unnecessary to include all components of the 
$(2,2)$-algorithm, for example when computing a second-order directional derivative.
In that case, we may define a restricted $(2,2)$-AD algorithm that drops the non-mixed second-order derivatives from the definitions above and defines $\Dsyn[(2,2)']{\reals}=\reals^4$ 
and\\ 
\resizebox{\linewidth}{!}{\parbox{1.025\linewidth}{$$\hspace{-5pt}\begin{array}{ll}
\Dsyn[(2,2)']{\op(\trm_1,...,\trm_n)}=&
\tMatch{\Dsyn[(2,2)']{\trm_1}}{\var^1_{00},\var^1_{01},\var^1_{10},\var^1_{11}}{}\\
&\vdots\\
&\tMatch{\Dsyn[(2,2)']{\trm_n}}{\var^n_{00},\var^n_{01},\var^n_{10},\var^n_{11}}{}\\
&\hspace{-4pt}\begin{array}{l}
\langle\op(\var^1_{00},\ldots,\var^n_{00}),\\
\sum_{i=1}^n\var^i_{01} *\partial_{\hat{i}}\op(\var_{00}^1,\ldots,\var_{00}^n),\\
\sum_{i=1}^n\var^i_{10} *\partial_{\hat{i}}\op(\var_{00}^1,\ldots,\var_{00}^n),\\
\sum_{i=1}^n\var^i_{11} *\partial_{{\hat{i}}}\op(\var_{00}^1,\ldots,\var_{00}^n)+
\sum_{i,j=1}^n\var^i_{10}*\var^j_{01}*\partial_{\widehat{{i,j}}}\op(\var_{00}^1,\ldots,\var_{00}^n)\rangle.\end{array}
\end{array}
$$}}
\end{example}

We extend $\Dsynsymbol$ to contexts: $\Dsyn{\{\var_1{:}\ty_1,{.}{.}{.},\var_n{:}\ty_n\}}\defeq
\{\var_1{:}\Dsyn{\ty_1},{.}{.}{.},\var_n{:}\Dsyn{\ty_n}\}$.
This turns $\Dsynsymbol$ into a well-typed, functorial macro in the following sense.
\begin{lem}[Functorial macro]\label{lem:functorialmacro}
	If $\ctx\tinf \trm:\ty$ then $\Dsyn{\ctx}\tinf \Dsyn{\trm}:\Dsyn{\ty}$.\\
	If $\ctx,\var:\ty[2]\tinf \trm:\ty$ and
	$\ctx\tinf\trm[2]:\ty[2]$ then
	$\Dsyn{\ctx}\tinf \Dsyn{\subst{\trm}{\sfor{\var}{\trm[2]}}}=\subst{\Dsyn{\trm}}{\sfor{\var}{\Dsyn{\trm[2]}}
	}$.
\end{lem}
\begin{proof}By induction on the structure of typing derivations.\end{proof}

\begin{example}[Inner products]\label{ex:innerprod}
Let us write $\ty^n$ for the $n$-fold product $\tProd{\ty}{\dots}{\ty}$.
Then, given $\Ginf{\trm,\trm[2]}{\reals^n}$, we can define their inner product as follows:\[
\begin{array}{ll}
\Gamma\vdash\trm\innerprod{n}\trm[2]\defeq\;&\tMatch{\trm}{\var[3]_1,\ldots,\var[3]_n}{}\\
&\tMatch{\trm[2]}{\var[2]_1,\ldots,\var[2]_n}{\var[3]_1 * \var[2]_1 + \dots + \var[3]_n * \var[2]_n}
:\reals	
\end{array}
\]
To illustrate the calculation of $\Dsynsymbol[(1,1)]$, let us expand (and $\beta$-reduce) $\Dsyn[(1,1)]{\trm\innerprod{2}\trm[2]}$:\begin{align*}
&\pMatch{\Dsyn[(1,1)]{\trm}}{\var[3]_1}{\var[3]_2}{}
\pMatch{\Dsyn[(1,1)]{\trm[2]}}{\var[2]_1}{\var[2]_2}{}\\
&\pMatch{\var[3]_1}{\var[3]_{1,1}}{\var[3]_{1,2}}{}
\pMatch{\var[2]_1}{\var[2]_{1,1}}{\var[2]_{1,2}}{}\\ 
&\pMatch{\var[3]_2}{\var[3]_{2,1}}{\var[3]_{2,2}}{}
\pMatch{\var[2]_2}{\var[2]_{2,1}}{\var[2]_{2,2}}{}\\
&\qquad \tPair{\var[3]_{1,1}*\var[2]_{1,1}+\var[3]_{2,1}*\var[2]_{2,1}\ }{\ \var[3]_{1,1}*\var[2]_{1,2}+\var[3]_{1,2}*\var[2]_{1,1}+\var[3]_{2,1}*\var[2]_{2,2}+\var[3]_{2,2}*\var[2]_{2,1}}
\intertext{Let us also expand the calculation of $\Dsyn[(2,2)']{\trm\innerprod{2}\trm[2]}$:}&\pMatch{\Dsyn[(2,2)']{\trm}}{\var[3]_1}{\var[3]_2}{}
\pMatch{\Dsyn[(2,2)']{\trm[2]}}{\var[2]_1}{\var[2]_2}{}\\
&\tMatch{\var[3]_1}{\var[3]_{1},\var[3]_{1,1}',\var[3]_{1,2}',\var[3]_{1}''}{}
\tMatch{\var[2]_1}{\var[2]_{1},\var[2]_{1,1}',\var[2]_{1,2}',\var[2]_{1}''}{}\\ 
&\tMatch{\var[3]_2}{\var[3]_{2},\var[3]_{2,1}',\var[3]_{2,2}',\var[3]_{2}''}{}
\tMatch{\var[2]_2}{\var[2]_{2},\var[2]_{2,1}',\var[2]_{2,2}',\var[2]_{2}''}{}\\
&\tTuple{\var[3]_{1}*\var[2]_{1}+\var[3]_{2}*\var[2]_{2},\\
&\qquad\ \var[3]_{1}*\var[2]_{1,1}'+\var[3]_{1,1}'*\var[2]_{1}+\var[3]_{2}*\var[2]_{2,1}'+\var[3]_{2,1}'*\var[2]_{2},\\
&\qquad\ \var[3]_{1}*\var[2]_{1,2}'+\var[3]_{1,2}'*\var[2]_{1}+\var[3]_{2}*\var[2]_{2,2}'+\var[3]_{2,2}'*\var[2]_{2},\\
&\qquad\ \var[3]_1''*\var[2]_1+\var[3]_2''*\var[2]_2+\var[2]_1''*\var[3]_1+\var[2]_2''*\var[3]_2 +\\
&\qquad\var[3]_{1,1}'*\var[2]_{1,2}'+\var[3]_{1,2}'*\var[2]_{1,1}'+\var[3]_{2,1}'*\var[2]_{2,2}'+\var[3]_{2,2}'*\var[2]_{2,1}'}
\end{align*}
\end{example}

\begin{example}[Neural networks]
In our introduction, we provided a program~\eqref{eqn:network} in our language 
to build a neural network from the expressions $\neuron,\layer,\compose$;
this program makes use of the inner product of Ex.~\ref{ex:innerprod}.
We can similarly calculate the derivatives of deep neural networks by applying the macro~$\Dsynsymbol$ mechanically.
\end{example}

\section{Semantics of differentiation}\label{sec:semantics}
Consider for a moment the first-order fragment of the language in Section~\ref{sec:simple-language}, with only one type, $\reals$, 
and no $\lambda$-abstractions or pairs. 
This has a simple semantics in the category of cartesian spaces and smooth maps.
Indeed, a term $\var_1\dots\var_n:\reals \vdash \trm:\reals$ has a natural reading
as a function $\sem{\trm}:\RR^n\to\RR$
by interpreting our operation symbols by the 
well-known operations on $\RR^n\to\RR$ with the corresponding name.
In fact, the functions that are definable in this first-order fragment are smooth.
Let us write $\CartSp$ for this category of cartesian spaces ($\RR^n$ for some $n$)
and smooth functions.

The category $\CartSp$ has cartesian products, and so we can also interpret product types, tupling and pattern matching,
giving us a useful syntax
for constructing functions into and out of products of $\RR$.
For example, the interpretation of $(\neuron_n)$ in (\ref{eqn:network})
becomes
\[
\RR^n\times \RR^n\times \RR \xto{\sem{\innerprod{n}}\times \id[\RR]}\RR\times \RR\xto{\sem{+}}\RR\xto{\sem{\sigmoid}}\RR.
\]
Here $\sem{\innerprod{n}}$, $\sem{+}$ and $\sem{\sigmoid}$ are the usual inner product, addition
and the sigmoid function on $\RR$, respectively.

Inside this category, we can straightforwardly study the first-order language without $\lambda$'s, and automatic differentiation.
In fact, we can prove the following by plain induction on the syntax:\\
\emph{The interpretation of the (syntactic) forward AD $\Dsyn{\trm}$ of a first-order term
$\trm$ equals the usual (semantic) derivative of the interpretation of $\trm$ as a smooth function.}

However, as is well-known, the category $\CartSp$ does not support function spaces. To see this, 
notice that we have polynomial terms 
\[\var_1,\ldots,\var_d:\reals\vdash \lambda \var[2].\,\textstyle\sum_{n=1}^d \var_n\var[2]^n:\reals\to\reals\]
for each $d$, and so if we could interpret $(\reals\to \reals)$ as a Euclidean space 
$\RR^p$ then, by interpreting these polynomial expressions, we would 
be able to find continuous injections $\RR^d\to \RR^p$ for every $d$, which is topologically impossible for any~$p$, for example as a consequence of the 
Borsuk-Ulam theorem (see \ifx\fossacsversion\undefined Appx.~\ref{sec:man_not_ccc}\else\cite{hsv-fossacs2020}, Appendix~A\fi).

This lack of function spaces means that we cannot interpret the functions $(\layer)$ and $(\compose)$ from~(\ref{eqn:network}) in $\CartSp$, as they are higher-order functions,
even though they are very useful building blocks for differential programming!
Clearly, we could define neural networks such as~(\ref{eqn:network}) directly as smooth functions 
without any higher-order subcomponents, though that would quickly become cumbersome for deep networks.
A problematic consequence of the lack of a semantics for higher-order differential programs is that we have no obvious way of establishing compositional semantic correctness of $\Dsynsymbol$ for the given implementation of~(\ref{eqn:network}).

We now show that every definable function is smooth, and then in Section~\ref{sec:simple-correctness} we show that the $\Dsynsymbol$ macro witnesses its derivatives. 

\subsection{Smoothness at higher types and diffeologies}
\newcommand{\setsem}[1]{\llfloor #1\rrceil}
The aim of this section is to introduce diffeological spaces as a semantic model for the simple language in Section~\ref{sec:simple-language}. By way of motivation, we begin with a standard set-theoretic semantics, where types are interpreted as follows
\[\textstyle \setsem \reals\defeq\RR \qquad\setsem{\tProd{\ty_1}{\dots}{\ty_n}}\defeq 
\prod_{i=1}^n\setsem{\ty_i}\qquad\setsem{\tau\to\sigma}\defeq(\setsem\tau\to\setsem\sigma)
\]
and a term ${x_1:\ty_1,\dots,x_n:\ty_n}\vdash t:\ty[2]$ is interpreted as a function $\prod_{i=1}^n\setsem {\ty_i}\to \setsem {\ty[2]}$, mapping a valuation of the context to a result.

We can show that the interpretation of a term $x_1:\reals,\dots,x_n:\reals\vdash t:\reals$ is always a smooth function $\RR^n\to \RR$, even if it has higher-order subterms. 
We begin with a fairly standard logical relations proof of this, and then move from it to the semantic model of diffeological spaces. 

\begin{prop}
  If $x_1:\reals,\dots,x_n:\reals\vdash t:\reals$ then the function $\setsem t: \RR^n\to\RR$ is smooth.\label{prop:smooth}
\end{prop}
\begin{proof}
For each type $\ty$ define a set $Q_{\ty}\subseteq [\RR^k\to \setsem {\ty}]$ by induction on the structure of types:
  \begin{align*}Q_\reals&=\{f:\RR^k\to \RR~|~f\text{ is smooth}\}\\
    Q_{\tProd{\ty_1}{\dots}{\ty_n}}&=\{\textstyle f:\RR^k\to\prod_{i=1}^n\setsem{\ty_i}~|~\forall i.\ (\lambda \vec r.\, f_i(\vec r))\in Q_{\ty_i}\}\\
    Q_{{\ty}\to{\ty[2]}}&=\{f:\RR^k\to\setsem\ty \to \setsem{\ty[2]}~|~\forall g\in Q_
    {\ty}.\, \lambda (\vec r).\,f(\vec r)(g(\vec r))\in Q_{\ty[2]}\}
  \end{align*}
  Now we show the fundamental lemma: if ${x_1:\ty_1,\dots,x_n:\ty_n}\vdash u:\ty[2]$
  and $g_1\in Q_{\ty_1},\dots,g_n\in Q_{\ty_n}$ then $((g_1\dots g_n);\setsem u)\in Q_{\ty[2]}$. This is shown by induction on the structure of typing derivations. The only interesting step here is that the basic operations ($+$, $*$, $\sigmoid$, etc.) are smooth. 
  We deduce the statement of the theorem by putting $u=t$, $k=n$, and letting $g_i:\RR^n\to\RR$ be the projections. 
\end{proof}
At higher types, the logical relations $Q$ show that we can only define functions that send smooth functions to smooth functions, meaning that we can never use them
to build first-order functions that are not smooth. For example, $(\compose)$ in~(\ref{eqn:network}) has this property.

This logical relations proof suggests building a semantic model by interpreting types as sets with structure: for each type we have a set $X$ together with a set $Q^{\RR^k}_X\subseteq [\RR^k\to X]$ of plots. 
\begin{defi}\label{def:diffeo}
	A \emph{diffeological space} $(X,\plots{X})$ consists of a set $X$ together with, for each $n$ and each open subset $U$ of $\RR^n$,  a set $\plots{X}^U\subseteq [U\to X]$ of functions, called \emph{plots}, such that
	\begin{itemize}
	 	\item all constant functions are plots;
	 	\item if $f:V\to U$ is a smooth function and $p\in\plots{X}^U$, then $f;p\in\plots{X}^V$;
     \item if $\seq[i\in I]{p_i\in\plots{X}^{U_i}}$ is a compatible family of plots $(x\in U_i\cap U_j\Rightarrow p_i(x)=p_j(x))$
     and $\seq[i\in I]{U_i}$ covers $U$,
     then the gluing $p:U\to X:x\in U_i\mapsto p_i(x)$ is a plot.
	 \end{itemize} 
\end{defi}
We call a function $f:X\to Y$ between diffeological spaces \emph{smooth} if, for all plots
$p\in\plots{X}^U$, we have $p;f\in \plots{Y}^U$. We write $\Diff(X,Y)$ for the set of smooth maps from $X$ to $Y$. 
Smooth functions compose, and so we have a category $\Diff$ of diffeological spaces and smooth functions.

A diffeological space is thus a set equipped with structure.
Many constructions of sets carry over straightforwardly to diffeological spaces.

\begin{example}[Cartesian diffeologies]\label{ex:cartesian-diffeologies}
Each open subset $U$ of $\RR^n$ can be given the structure of a diffeological space by taking all the
smooth functions $V\to U$ as $\plots{U}^V$.
Smooth functions from $V\to U$ in the traditional sense coincide with
smooth functions in the sense of diffeological spaces~\cite{iglesias2013diffeology}.
Thus diffeological spaces have a profound relationship with ordinary calculus.

In categorical terms, this gives a full embedding of $\CartSp$ in $\Diff$. 
\end{example}
\begin{example}[Product diffeologies]
Given a family $\seq[i\in I]{X_i}$ of diffeological spaces,
we can equip the product $\prod_{i\in I}X_i$ of sets with the
\emph{product diffeology} in which $U$-plots are precisely the functions
of the form $\seq[i\in I]{p_i}$ for $p_i\in\plots{X_i}^U$.  

This gives us the categorical product in $\Diff$.
\end{example}
\begin{example}[Functional diffeology]
We can equip the set $\Diff(X,Y)$ of smooth functions between diffeological spaces with the \emph{functional diffeology}
in which $U$-plots consist of functions $f:U\to \Diff(X,Y)$ such that 
$(u,x)\mapsto f(u)(x)$ is an element of $\Diff(U\times X, Y)$.

This specifies the categorical function object in $\Diff$.
\end{example}

We can now give a denotational semantics for our language from Section~\ref{sec:simple-language} in the category of diffeological spaces. 
We interpret each type $\ty$ as a set $\sem \ty$ equipped with the relevant diffeology,
by induction on the structure of types:

\begin{align*}
    & \sem \reals \defeq \RR\qquad\text{with the standard diffeology} \\
    & \sem{\tProd{\ty_1}{\dots}{\ty_n}}\ \defeq\ \textstyle\prod_{i=1}^n\sem{\ty_i}\quad\text{with the product diffeology} \\
    & \sem{\ty\To\ty[2]} \defeq \Diff(\sem \ty,\sem{\ty[2]})\quad\text{with the functional diffeology}
\end{align*}
A context $\Gamma=(\var_1\colon\ty_1\dots \var_n\colon \ty_n)$ is interpreted as a diffeological space
$\sem \Gamma\defeq \prod_{i=1}^n\sem{\ty_i}$. 
Now, well-typed terms $\Gamma\vdash \trm:\ty$ are interpreted as smooth functions
$\sem \trm:\sem\Gamma\to \sem \ty$, giving a meaning to $\trm$ for every valuation of the context. 
This is routinely defined by induction on the structure of typing derivations once we choose 
a smooth function $\sem{\op}:\RR^n\to \RR$ to interpret each $n$-ary operation $\op\in\Op_n$. 
For example, constants $\cnst:\reals$ are interpreted as constant functions;
and the first-order operations ($+,*,\sigmoid$) are interpreted by composing with the corresponding functions, which are smooth: e.g., $\sem{\sigmoid(t)}(\rho)\defeq\sigmoid(\sem t(\rho))$, where $\rho\in\sem\Gamma$. 
Variables are interpreted as $\sem{\var_i}(\rho) \defeq \rho_i$. 
The remaining constructs are interpreted as follows, and it is straightforward to show that smoothness is preserved. 
\begin{align*}
&
\sem{\tTriple {\trm_1}{\dots}{ \trm_n}}(\rho)\defeq
(\sem{\trm_1}(\rho),\dots,\sem{\trm_n}(\rho))
&&
\sem{\fun{\var{:}\ty}{\trm}}(\rho)(a)\defeq
\sem {\trm}(\rho,a)\ \text{($a\in \sem {\ty}$)}
\\
&
\sem{\tMatch{\trm}{{.}{.}{.} }{\trm[2]}}(\rho)\defeq
\sem{\trm[2]}(\rho,\sem{\trm}(\rho))
&&
\sem{\trm\,\trm[2]}(\rho)\defeq
\sem{\trm}(\rho)(\sem{\trm[2]}(\rho))
\end{align*}

The logical relations proof of Proposition~\ref{prop:smooth} is reminiscent of diffeological spaces. We now briefly remark on the suitability of the axioms of diffeological spaces (Definition~\ref{def:diffeo}) for a semantic model of smooth programs. The first axiom says that we only consider reflexive logical relations. From the perspective of the interpretation, it recognizes in particular that the semantics of an expression of type $(\reals\to\reals)\to\reals$ is defined by its value on smooth functions rather than arbitrary arguments. That is to say, the set-theoretic semantics at the beginning of this section, $\setsem{(\reals\to\reals)\to\reals}$, is different from the diffeological semantics, $\sem{(\reals\to\reals)\to\reals}$. The second axiom for diffeological spaces ensures that the smooth maps in $\Diff(U,X)$ are exactly the plots in $\plots X^U$. The third axiom ensures that categories of manifolds fully embed into $\Diff$; it will not play a visible role in this paper -- in fact, Barthe et al.~\cite{bcdg-open-logical-relations} prove similar 
results for a simple language like ours by using plain logical relations (over $\Set$) and without demanding the diffeology axioms. However, we expect the third axiom to be crucial for programming with other smooth structures or partiality.

\subsection{Correctness of AD}\label{sec:simple-correctness}
We have shown that a term
$\var_1\colon\reals,\dots,\var_n\colon\reals\vdash \trm : \reals$
is interpreted as a smooth function
$\sem \trm:\RR^n\to \RR$, even if $\trm$ involves higher-order functions (like~(\ref{eqn:network})).
Moreover, the macro translation $\Dsyn[(k,R)]\trm$ is a function 
$\sem {\Dsyn[(k,R)] \trm}:(\RR^{\cchoose{R+k}{k}})^n\to \RR^{\cchoose{R+k}{k}}$ (Lemma~\ref{lem:functorialmacro}). 
This enables us to state a limited version of our main correctness theorem:
\begin{thm}[Semantic correctness of $\Dsynsymbol$ (limited)]
  \label{thm:fwd-cor-basic}
  For any term\\ $\var_1\colon\reals,\dots,\var_n\colon\reals\vdash \trm : \reals$, the function
  $\sem {\Dsyn[(k,R)] \trm}$ is the $(k,R)$-Taylor representation \eqref{eqn:taylorrepresentation} of
  $\sem \trm$.
  In detail: for any smooth functions
  $f_1\dots f_n:\RR^k\to\RR$, 
  \[
    \resizebox{\linewidth}{!}{\parbox{1.2\linewidth}{\[
    {\left({\left(
      \frac{\partial^{\alpha_1+\ldots+\alpha_k}f_j(x)}{\partial x_1^{\alpha_1}\cdots \partial x_k^{\alpha_k}}\right)}_{(\alpha_1,...,\alpha_k)=(0,...,0)}^{(R,0,...,0)}\right)}_{j=1}^n\hspace{-10pt};\sem{\Dsyn[(k,R)]\trm}
    =
    {\left(
      \frac{\partial^{\alpha_1+\ldots+\alpha_k}((f_1,\ldots,f_n);\sem{\trm})(x)}{\partial x_1^{\alpha_1}\cdots \partial x_k^{\alpha_k}}\right)}_{(\alpha_1,...,\alpha_k)=(0,...,0)}^{(R,0,...,0)}\hspace{-12pt}
    \text.\]}}
    \]
\end{thm}
For instance, if $n=2$, then 
$\sem{\Dsyn[(1,1)]\trm}(\var_1,1,\var_2,0)=
\big(\sem\trm(\var_1,\var_2),\frac{\partial\sem\trm(\var,\var_2)}{\partial \var}(\var_1)\big)$.
\begin{proof}
  We prove this by logical relations. 
  A categorical version of this proof is in Section~\ref{sec:gluing}. 

For each type $\ty$, we define a binary relation
 $S_{\ty}$ between (open) $k$-dimensional plots in $\sem{\ty}$ and  (open) $k$-dimensional plots in $\sem{\Dsyn[(k,R)]{\ty}}$,
 i.e.~$S_{\ty}\subseteq  \plots{\sem{\ty}}^{\RR^k}\times
 \plots{\sem{\Dsyn[(k,R)]{\ty}}}^{\RR^k}$,
 by induction on $\ty$:
\begin{align*}
&S_{\reals}&\defeq&\left\{\left(f,{\left(
  \frac{\partial^{\alpha_1+\ldots+\alpha_k}f(x)}{\partial x_1^{\alpha_1}\cdots \partial x_k^{\alpha_k}}\right)}_{(\alpha_1,...,\alpha_k)=(0,...,0)}^{(R,0,...,0)}\right)~\Big|~f:\RR^k\to\RR\text{ smooth}\right\} \\
& S_{\tProd{\ty_1}{...}{\ty_n}} &\defeq&\{\sPair{\sTuple{f_1,...,f_n}}{\sTuple{g_1,...,g_n}}\mid
(f_1,g_1)\in S_{\ty_1},..., (f_n,g_n)\in S_{\ty_n}\} \\
& S_{\ty\To \ty[2]} &\defeq &\{(f_1,f_2) \mid \forall (g_1,g_2) \in S_{\ty}. (x{\mapsto} f_1(x)(g_1(x)),x{\mapsto} f_2(x)(g_2(x))) \in S_{\ty[2]} \} \\
\end{align*}

\noindent We then establish the following `fundamental lemma':

\begin{quotation}
\noindent If $\var_1{:} \ty_1,{.}{.}{.}, \var_n {:} \ty_n \vdash \trm : \ty[2]$
and, for all $1{\leq} i{\leq} n$,
$f_i: \RR^k\to \sem{\ty_i}$ and\\ $g_i : \RR^k\to~\sem{\Dsyn[(k,R)]{\ty_i}}$ are such that $(f_i,g_i)$ is in $S_{\ty_i}$,
then
\[\Big((f_1,\ldots,f_n);\sem{\trm},(g_1,\ldots,g_n);\sem{\Dsyn[(k,R)]{\trm}}\Big)\]
is in $S_{\ty[2]}$.
\end{quotation}

This is proved by a routine induction on the typing derivation of $\trm$.
The case for $\op(\trm_1,\ldots,\trm_n)$ relies on the precise definition of $\Dsyn[(k,R)]{\op(\trm_1,\ldots,\trm_n)}$. 

We conclude the theorem from the fundamental lemma by considering the case where $\ty_i=\ty[2]=\reals$
and taking each $g_i$ to be the tuple of derivatives of $f_i$ appearing on the left-hand side.
\end{proof}

\section{Extending the language: variant and inductive types}\label{sec:extended-language}
In this section, we show that the definition of forward AD and the semantics generalize
if we extend the language of Section~\ref{sec:simple-language} with variants
and inductive types.
As an example of inductive types, we consider lists.
This specific choice is only for expository purposes, and the whole
development works at the level of generality of arbitrary algebraic
data types generated as initial algebras of (polynomial) type constructors formed by
finite products and variants.
These types are easily interpreted in the category of diffeological spaces in much the same way. 
The categorically minded reader may regard this as a consequence of $\Diff$ being a concrete Grothendieck quasitopos, 
e.g.~\cite{baez2011convenient}, and hence complete and cocomplete. 

\subsection{Language.}\label{sec:extended-language-language}
We additionally consider the following types and terms:

\begin{figure}[H]
  \noindent\begin{syntax}
    \ty, \ty[2], \ty[3] & \gdefinedby & \ \dots & \syncat{types}                          \\
    &\gor& \Variant{
        \Inj{\Cns_1}{\ty_1}
        \vor \ldots \vor
        \Inj{\Cns_n}{\ty_n}
      }        &\synname{variant}          \\
&\gor\quad\,& \List{\ty}                 & \synname{list}\\[6pt]
    \trm, \trm[2], \trm[3] & \gdefinedby & \ \dots & \syncat{terms}                          \\
    &\gor&\tInj\ty\Cns\trm               & \synname{variant constructor}             \\
    &\gor& \tNil
    \ \gor\ \tCons{\trm}{\trm[2]}          & \synname{empty list and cons}\\
    &\gor& \vMatch {\trm  }  {
                        \Inj{\Cns_1}{\var_1}\To{\trm[2]_1}
                \vor \cdots
                \vor  \Inj{\Cns_n}{\var_n}\To{\trm[2]_n}
                }           & \synname{pattern matching: variants}\\
    &\gor& \lFold{\var_1}{\var_2}{\trm}{\trm[2]}{\trm[3]} & \synname{list fold}\\
\end{syntax} \end{figure}
We extend the type system according to the rules of Fig.~\ref{fig:types2}.
\begin{figure}[b]
  \oframebox{\scalebox{1}{\begin{minipage}{\linewidth}\[
  \begin{array}{@{}c@{}}
  \inferrule{
    \Ginf\trm{\ty_i}
  }{
    \Ginf{\tInj\ty{\Cns_i}\trm}{\ty}
  }((\Inj{\Cns_i}\ty_i) \in \ty)
\quad
  \inferrule{
    ~
  }{
    \Ginf \tNil {\List{\ty}} 
  }
  \quad
  \inferrule{
  \Ginf \trm \ty
  \\
  \Ginf {\trm[2]} {\List{\ty}}
  }{
  \Ginf {\tCons{\trm}{\trm[2]}} {\List{\ty}}
  }
\\
\\
  \inferrule{
    \Ginf\trm{\Variant{
                \Inj{\Cns_1}{\ty_1}
                \vor \ldots \vor
                \Inj{\Cns_n}{\ty_n}}}
    \\
    \text{for each $1 \leq i \leq n$: }
    \Ginf[, \var_i : \ty_i]{\trm[2]_i}{\ty}
  }{
    \Ginf{\vMatch \trm
                {\begin{array}[t]{@{}l@{\,}l@{}l@{}}
                    \Inj{\Cns_1}{\var_1}\To{\trm[2]_1}
                    \vor\cdots
                    \vor\Inj{\Cns_n}{\var_n}&\To{\trm[2]_n}
    }}
    \ty
                  \end{array}
  }
  \\
  \\
  \inferrule{
    \Ginf {\trm[2]} {\List{\ty}}
    \\
    \Ginf {\trm[3]} {\ty[2]}
    \\
    \Ginf[{,\var_1:\ty,\var_2:\ty[2]}] {\trm} {\ty[2]}
    }{
    \Ginf {\lFold{\var_1}{\var_2}{\trm}{\trm[2]}{\trm[3]}} {\ty[2]}
    }
\end{array}
\]
 \end{minipage}}}
\caption{Additional typing rules for the extended language.\label{fig:types2}}
\end{figure}
We can then extend $\Dsynsymbol[(k,R)]$ (again, writing it as $\Dsynsymbol$, for legibility) to our new types and terms by
\begin{align*}
        &\Dsyn{\Variant{\Inj{\Cns_1}{\ty_1}\vor \ldots \vor\Inj{\Cns_n}{\ty_n}}} \defeq 
        \Variant{\Inj{\Cns_1}{\Dsyn{\ty_1}}\vor \ldots \vor\Inj{\Cns_n}{\Dsyn{\ty_n}}}\\
        &\Dsyn{\List{\ty}} \defeq \List{\Dsyn{\ty}}\\
 &\Dsyn{\tInj\ty\Cns\trm} \defeq \tInj{\Dsyn{\ty}}\Cns{\Dsyn\trm}\\
& \Dsyn{\tNil} \defeq \tNil \\
& \Dsyn{\tCons{\trm}{\trm[2]}} \defeq \tCons{\Dsyn{\trm}}{\Dsyn{\trm[2]}}\\
&\Dsyn{\vMatch {\trm  }  {
    \Inj{\Cns_1}{\var_1}\To{\trm[2]_1}
\vor \cdots
\vor  \Inj{\Cns_n}{\var_n}\To{\trm[2]_n}
}} \defeq\\
&\,\quad\vMatch {\Dsyn\trm  }  {
    \Inj{\Cns_1}{\var_1}\To{\Dsyn{\trm[2]_1}}
\vor \cdots
\vor  \Inj{\Cns_n}{\var_n}\To{\Dsyn{\trm[2]_n}}
} \\
&\Dsyn{\lFold{\var_1}{\var_2}{\trm}{\trm[2]}{\trm[3]}} \defeq
\lFold{\var_1}{\var_2}{\Dsyn\trm}{\Dsyn{\trm[2]}}{\Dsyn{\trm[3]}}
\end{align*} 
To demonstrate the practical use of expressive type systems for
differential programming, we consider the following two examples.
\begin{example}[Lists of inputs for neural networks]
Typically, we run a neural network on a large dataset, the size of 
which might be determined at runtime.
In practice, one often evaluates a neural network on multiple inputs by summing the outcomes.
This can be written in our extended language as follows.
Suppose that we have a network $f:\bProd{\reals^n}{P}\To\reals$ that operates on single input vectors. 
We can construct one that operates on lists of inputs as follows:
\[
g\defeq \fun{\tTuple{l,w}}{\lFold{\var_1}{\var_2}{f\tTuple{\var_1,w} + \var_2}{l}{\underline{0}}} :
\bProd{\List{\reals^n}}{P}\To\reals
\]
\end{example}

\begin{example}[Missing data]
In practically every application of statistics and machine learning,
we face the problem of \emph{missing data}:
for some observations, only partial information is available.

In an expressive typed programming language such as the one we consider,
we can model missing data conveniently using the data type
$\Maybe{\ty}=\Variant{
  \Inj{\tNothingSym}{\Unit}\vor
  \Inj{\tJustSym}{\ty}
} $.
In the context of a neural network, one might use it as follows.
First, define some helper functions

\[
\begin{aligned}
&\tFromMaybe{\ty} \defeq\fun{\var}{\fun{m}{
  \vMatch {m  }  {
    \Inj{\tNothingSym}{\_}\To{\var}
\vor  \Inj{\tJustSym}{\var'}\To{\var'}
}  
}}
\\&
\tFromMayben{\ty}{n}\defeq
\fun{\tTriple{\var_1}{{.}{.}{.}}{\var_n}}{\fun{\tTriple{m_1}{{.}{.}{.}}{m_n}}{\tTriple{\tFromMaybe{\ty}\,\var_1\,m_1}{{.}{.}{.}}{\tFromMaybe{\ty}\,\var_n\,m_n}}  }\\&\qquad\qquad:\ty^n\To (\Maybe{\ty})^n\To\ty^n\\&
\tMap{\ty}{\ty[2]}\defeq\fun{f}{\fun{l}{\lFold{\var_1}{\var_2}{\tCons{f\, \var_1}{\var_2}}{l}{\tNil}}}
: (\ty\To\ty[2])\To\List{\ty}\To\List{\ty[2]}
\end{aligned}
\]
Given a neural network $f:\bProd{\List{\reals^k}}{P}\To\reals$,
we can build a new one that operates on 
a dataset for which some covariates (features) are missing, by passing 
default values to replace the missing covariates:
\begin{multline*}
\fun{\tTuple{l,\tTuple{m,w}}}
f\tTuple{\tMap{\reals}{\reals}\, (\tFromMayben{\reals}{k}\,m)\, l
,w}
:\bProd{\List{(\Maybe{\reals})^k}}{\bProd{\reals^{k}}{P}}\To\reals\end{multline*}
Then, given a dataset $l$ with missing covariates, we can perform automatic differentiation on this network to optimize the ordinary network parameters~$w$ \emph{and} the default values for missing covariates $m$ simultaneously.
\end{example}

\subsection{Semantics.}\label{sec:extended-language-semantics}
In Section~\ref{sec:semantics} we gave a denotational semantics for the simple language in diffeological spaces. This extends to the language in this section, as follows.
As before, each type $\ty$ is interpreted as a diffeological space, which is a set equipped with a family of plots:
\begin{itemize}
\item A variant type $\Variant{
        \Inj{\Cns_1}{\ty_1}
        \vor \ldots \vor
        \Inj{\Cns_n}{\ty_n}
      } $ is inductively interpreted as the disjoint union of the semantic spaces,
      $\textstyle\sem{\Variant{
        \Inj{\Cns_1}{\ty_1}
        \vor \dots \vor
        \Inj{\Cns_n}{\ty_n}
      }} \ \ \defeq \ \ \biguplus_{i=1}^n\sem {\ty_i}$, with $U$-plots
  \[  \textstyle\plots{
    \sem{\Variant{
        \Inj{\Cns_1}{\ty_1}
        \vor \ldots \vor
        \Inj{\Cns_n}{\ty_n}
      }}
    }^U\hspace{-4pt}\defeq
    \left\{\left.\coseq[j= 1]{U_j\xto {f_j}\sem{\ty_j}\to\biguplus_{i=1}^n\sem{\ty_i}}^n\hspace{-4pt}~\right|~U=\biguplus_{j=1}^n U_j,\;f_j\in\plots{\sem{\ty_j}}^{U_j}\right\}.\]
  \item A list type $\List\ty$ is interpreted as the union of the sets of length $i$ tuples for all natural numbers $i$,
    $ \sem{\List\ty} \ \ \defeq\ \  \biguplus_{i=0}^\infty \sem\ty^i$
    with $U$-plots
    \[\textstyle\plots{
      \sem{\List\ty}
      }^U\defeq\left\{\left.\coseq[j=0]{U_j\xto {f_j}\sem{\ty}^j\to\biguplus_{i=0}^\infty\sem{\ty}^i}^\infty~\right|~U=\biguplus_{j=0}^\infty U_j,\;f_j\in\plots{\sem{\ty}^j}^{U_j}\right\}\]
  \end{itemize}
  The constructors and destructors for variants and lists are interpreted as
  in the usual set-theoretic semantics.

  It is routine to show inductively that these interpretations are smooth. Thus every term
  $\Gamma\vdash \trm:\ty$ in the extended language is interpreted as a smooth function
  $\sem \trm:\sem\Gamma\to\sem \ty$ between diffeological spaces. 
  List objects as initial algebras are computed as usual in a cocomplete category (e.g.~\cite{jacobsrutten2011}). 
  More generally, the interpretation for algebraic data types follows exactly the usual categorical semantics of variant types and inductive types (e.g.~\cite{pitts1995categorical}).

\section{Categorical analysis of (higher-order) forward AD and its correctness}\label{sec:correctness}
This section has three parts. First, we give a categorical account of the functoriality of AD (Ex.~\ref{ex:canonical-fwd}). Then we introduce our gluing construction, and relate it to the correctness of AD (diagram~\ref{dgm:gluing}).
Finally, we state and prove a correctness theorem for all first-order types by considering a category of manifolds~(Theorem~\ref{thm:fwd-cor-full}). 

\subsection{Syntactic categories.}
The key contribution of this subsection is that the AD macro translation (Section~\ref{sec:admacro})
has a canonical status as a unique functor between categories with structure.
To this end, we build a syntactic category $\Syn$ from our language, which is a \emph{free} category with certain structure.
This means that for any category~$\catC$ with this structure, there is a unique structure-preserving functor $\Syn\to \catC$,
which is an interpretation of our language in that category. Generally speaking, this is the categorical view of denotational semantics (e.g.~\cite{pitts1995categorical}).
In this particular setting, the category $\Syn$ itself admits alternative forms of this structure, given by the dual numbers interpretation, the triple numbers interpretation, etc. of Section~\ref{sec:dual-numbers-taylor}.
This gives canonical functors $\Syn\to \Syn$ translating the language into itself, which are the AD macro translations (Section~\ref{sec:admacro}).
A key point is that $\Syn$ is almost entirely determined by universal properties
(for example, cartesian closure for the function space); the only freedom is in the choice of interpretation of
\begin{enumerate}
  \item the real numbers $\reals$, which can be taken as the plain type $\reals$, or as the dual numbers interpretation $\reals \ast\reals$ etc.;
  \item the primitive operations $\op$, which can be taken as the operation $\op$ itself, or as the derivative of the operation, etc.
\end{enumerate}

\begin{figure}[b]
  \oframebox{\scalebox{1.0}{\begin{minipage}{\linewidth}\begin{align*}
    &\tMatch{\tTriple{\trm_1}{\ldots}{\trm_n}}{\var_1,\ldots,\var_n}{\trm[2]}=
  \subst{\trm[2]}{\sfor{\var_1}{\trm_1},\ldots,\sfor{\var_n}{\trm_n}}
  \\
  &\subst{\trm[2]}{\sfor{\var[2]}{\trm}}\freeeq{\var_1,\ldots,\var_n}
\tMatch{\trm}{\var_1,\ldots,\var_n}{\subst{\trm[2]}{
    \sfor{\var[2]}{\tTriple{\var_1}{\ldots}{\var_n}}}}
\\
&\vMatch {\Inj{\Cns_i}{\trm}  }  {
        \Inj{\Cns_1}{\var_1}\To{\trm[2]_1}
\vor \cdots
\vor  \Inj{\Cns_n}{\var_n}\To{\trm[2]_n}
} = \subst{\trm[2]_i}{\sfor{\var_i}{\trm}}\\
&\subst{\trm[2]}{\sfor{\var[2]}{\trm}}\freeeq{\var_1,\ldots,\var_n}
\vMatch {\trm}  {
        \Inj{\Cns_1}{\var_1}\To{\subst{\trm[2]}{\sfor{\var[2]}{\Inj{\Cns_1}{\var_1}}}}
\vor \cdots
\vor  \Inj{\Cns_n}{\var_n}\To{\subst{\trm[2]}{\sfor{\var[2]}{\Inj{\Cns_n}{\var_n}}}}
}\\
&\lFold{\var_1}{\var_2}{\trm}{\tNil}{\trm[3]} = \trm[3]
\\
&\lFold{\var_1}{\var_2}{\trm}{\tCons{\trm[2]_1}{\trm[2]_2}}{\trm[3]} =
\subst{\trm}{\sfor{\var_1}{\trm[2]_1},\sfor{\var_2}
{\lFold{\var_1}{\var_2}{\trm}{\trm[2]_2}{\trm[3]}}}
 \\
&u= \subst{\trm[2]}{\sfor{\var[2]}{\tNil}} ,
\subst{\trm[3]}{\sfor{\var_2}{\trm[2]}}
=\subst{\trm[2]}{\sfor{\var[2]}{\tCons{\var_1}{\var[2]}}}
\Rightarrow
 \subst{\trm[2]}{\sfor{\var[2]}{\trm}} \freeeq{\var_1,\var_2}
\lFold{\var_1}{\var_2}{\trm[3]}{\trm}{u} \\
&(\fun{\var}{\trm})\,\trm[2] = \subst{\trm}{\sfor{\var}{\trm[2]}}\\
&\trm \freeeq{\var} \fun{\var}{\trm\,\var} \\
&\quad \text{\it We write $\freeeq{\var_1,\ldots,\var_n}$ to indicate that the variables are not free in the left-hand side}
\end{align*} \end{minipage}}}
\caption{Standard $\beta\eta$-laws (e.g.~\cite{pitts1995categorical}) for products, functions, variants and lists. \label{fig:beta-eta}}
\end{figure}
In more detail, our language induces a syntactic category as follows.
\begin{defi}
	Let $\Syn$ be the category whose objects are types,
        and where a morphism $\ty[1]\to\ty[2]$ is a term in context 
        $\var:\ty[1]\vdash \trm:\ty[2]$ modulo the $\beta\eta$-laws
        (Fig.~\ref{fig:beta-eta}).
        Composition is by substitution. 
\end{defi}
For simplicity, we do not impose identities involving the primitive operations, such as the arithmetic identity $x+y=y+x$ in $\Syn$.
As is standard, this category has the following universal property.

\begin{lem}[e.g.~\cite{pitts1995categorical}]\label{lem:syn-initial}
  For every bicartesian closed category $\catC$ with list objects,
  and every choice of an object $\freeF(\reals)\in\catC$ and
  morphisms $\freeF(\op)\in\catC(\freeF(\reals)^n, \freeF(\reals))$ for all $\op\in \Op_n$ and $n\in\NN$,
  in $\catC$, there is a unique functor $\freeF:{\Syn\to\catC}$ respecting these interpretations and preserving the bicartesian closed structure as well as
 list objects.
\end{lem}
\begin{proof}[Proof notes]
  The functor $\freeF:\Syn\to \catC$ is a canonical denotational semantics
  for the language, interpreting types as objects of $\catC$ and terms as morphisms. 
  For instance,
  $\freeF({\ty\To\ty[2]})\defeq (\freeF\ty\To\freeF{\ty[2]})$,
  the function space in the category $\catC$,
  and $\freeF{(\trm\,\trm[2])}$ 
  is the composite $(\freeF{\trm},\freeF{\trm[2]});\mathit{eval}$. 
\end{proof}
When $\catC=\Diff$,
the denotational semantics of the language in diffeological spaces (Sections~\ref{sec:semantics} and~\ref{sec:extended-language-semantics})
can be understood as the unique structure-preserving functor 
$\sem-:\Syn\to \Diff$ satisfying 
$\sem \reals=\RR$, $\sem \sigmoid=\sigmoid$ and so on. 
  
\begin{example}[Canonical definition of forward AD]\label{ex:canonical-fwd}
The forward AD macro $\Dsynsymbol[(k,R)]$ (Sections~\ref{sec:simple-language} and~\ref{sec:extended-language-language}) 
arises as a canonical bicartesian closed functor on $\Syn$ that preserves list objects. 
Consider the unique
bicartesian closed functor $\freeF:\Syn\to\Syn$ that preserves list objects
such that $\freeF(\reals)=\reals^{\cchoose{R+k}k}$ and
$$\freeF(\op)=\var[3]\!:\!\tProd{\freeF(\reals)}{\ldots}{\freeF(\reals)}\vdash 
\tMatch{\var[3]}{\var_1,...,\var_n}{\Dsyn[\!(k,R)]{\op(\var_1,\ldots,\var_n)}}:\freeF(\reals).$$
  Then, for any type $\ty$, $\freeF(\ty)=\Dsyn[(k,R)] \ty$, and 
  for any term $x:\ty[1]\vdash \trm:\ty[2]$, 
  $\freeF(\trm)=\Dsyn[(k,R)] \trm$ as morphisms $\freeF(\ty[1])\to \freeF(\ty[2])$ in the syntactic category.

  This observation is a categorical counterpart to Lemma~\ref{lem:functorialmacro}.
\end{example}

\subsection{Categorical gluing and logical relations.}\label{sec:gluing}
Gluing is a method for building new categorical models and has been used for many purposes, including logical relations and realizability~\cite{mitchell1992notes}. 
Our logical relations argument in the proof of Theorem~\ref{thm:fwd-cor-basic} can be understood in this setting.
(In fact we originally found the proof of Theorem~\ref{thm:fwd-cor-basic} in this way.) 
In this subsection, for the categorically minded, we explain this, and in doing so we quickly recover a correctness result for the more general language in Section~\ref{sec:extended-language} and 
for arbitrary first-order types.

The general, established idea of categorical logical relations starts from the observation that logical relations are defined by induction on the structure of types.
Types have universal properties in a categorical semantics (e.g.~cartesian closure for the function space), and so we can organize the logical relations argument by defining some category~$\catC$ of relations and observing that it has the requisite categorical structure. The interpretation of types as relations can then be understood as coming from a unique structure-preserving map $\Syn\to\catC$. 
In this paper, our logical relations are not quite as simple as a binary relation on sets; rather, they are relations between plots. Nonetheless, these relations still form a category with the appropriate structure, which follows because they can still be regarded as arising from a gluing construction, as we now explain. 

We define a category $\Gl[k]$ whose objects are triples $(X,X',S)$ where $X$ and~$X'$ are diffeological spaces and $S\subseteq\plots X^{\RR^k}\times \plots {X'}^{\RR^k}$ is a relation between their $k$-dimensional plots. A morphism $(X,X',S)\to (Y,Y',T)$ is a pair of smooth functions $f\colon X\to Y$, $f'\colon X'\to Y'$, such that if 
$(g,g')\in S$ then $(g;f,g';f')\in T$.
The idea is that this is a semantic domain in which we can simultaneously interpret the language and its automatic derivatives.
\begin{prop}\label{prop:gluing}
  The category $\Gl[k]$ is bicartesian closed, has list objects, and the projection functor 
  $\projf:\Gl[k]\to\Diff\times \Diff$ preserves this structure. 
\end{prop}
\begin{proof}[Proof notes]
The category $\Gl[k]$ is a full subcategory of the comma category\\ $\id[\Set]\downarrow \Diff(\RR^k,-)\times \Diff(\RR^k,-)$. 
The result thus follows by the general theory of categorical gluing~(e.g.~\cite[Lemma~15]{johnstone-lack-sobocinski}).

\end{proof}
We give a semantics $\semgl{-}=(\semgl{-}_0,\semgl{-}_1, S_{-})$ for the language in $\Gl[k]$, 
interpreting types $\ty$ as objects $(\semgl{ \ty}_0,\semgl{\ty}_1,S_{\ty})$,
and terms as morphisms. 
We let $\semgl{\reals}_0\defeq \RR$ and 
$\semgl{\reals}_1\defeq \RR^{\cchoose{R+k}k}$, with the relation
$$S_\reals\defeq \left\{(f,{\left(
  \frac{\partial^{\alpha_1+\ldots+\alpha_k}f(x)}{\partial x_1^{\alpha_1}\cdots \partial x_k^{\alpha_k}}\right)}_{(\alpha_1,...,\alpha_k)=(0,...,0)}^{(R,0,...,0)})~|~f:\RR^k\to\RR\text{ smooth}\right\}.$$
We interpret the operations $\op$ according to $\sem{\op}$ in $\semgl{-}_0$, but according to the $(k,R)$-Taylor representation of $\sem{\op}$ in $\semgl{-}_1$.
For instance, when $k=2$ and $R=2$,
$\semgl{*}_1:\RR^6\times\RR^6\to \RR^6$ is 
\begin{align*}\semgl{*}_1&((x_{00},x_{01},x_{02},x_{10},x_{11},x_{20}),
(y_{00},y_{01},y_{02},y_{10},y_{11},y_{20}))\defeq\\
&(x_{00}y_{00},\\
&\;x_{00}y_{01}+x_{01}y_{00},\\
&\;x_{02}y_{00}+2x_{01}y_{01}+x_{00}y_{02},\\
&\;x_{00}y_{10}+x_{10}y_{00},\\
&\;x_{11}y_{00}+x_{01}y_{10}+x_{10}y_{01}+x_{00}y_{11},\\
&\;x_{20}y_{00}+2x_{10}y_{10}+x_{00}y_{20})\text.\end{align*}
At this point, one checks that these interpretations are indeed morphisms in $\Gl[k]$. 
This is equivalent to the statement that $\semgl{\op}_1$ is the $(k,R)$-Taylor representation of 
$\sem{\op}$~(\ref{eqn:taylorrepresentation}). 
The remaining constructions of the language are interpreted using the categorical structure of $\Gl[k]$, following Lemma~\ref{lem:syn-initial}.

Notice that the following diagram commutes. One can check this by hand or note that it follows from the initiality of $\Syn$ (Lemma~\ref{lem:syn-initial}):
all the functors preserve all the structure. 
\begin{equation}\label{dgm:gluing}
\xymatrix{
\Syn\ar[rr]^-{(\id,\Dsyn[(k,R)]-)}\ar[d]_{\semgl-}&&\Syn\times \Syn\ar[d]^{\sem-\times \sem-}
\\
\Gl[k]\ar[rr]_-{\projf}&&\Diff\times\Diff
}
\end{equation}
We thus arrive at a restatement of the correctness theorem (Theorem~\ref{thm:fwd-cor-basic}), which holds even for the extended language with variants and lists, because 
for any $x_1:\reals,{.}{.}{.},x_n:\reals\vdash \trm:\reals$, 
the interpretations $(\sem \trm,\sem{\Dsyn[(k,R)] \trm})$ are in the image of the projection $\Gl[k]\to \Diff\times\Diff$, 
and hence $\sem{\Dsyn[(k,R)]\trm}$ is a $(k,R)$-Taylor representation of~$\sem\trm$.

\subsection{Correctness at all first-order types, via manifolds.}
We now generalize Theorem~\ref{thm:fwd-cor-basic} to hold at all first-order types, not just the reals.

So far, we have shown that our macro translation (Section~\ref{sec:admacro}) gives correct derivatives for functions of real numbers, even if other types are involved in the definitions of the functions (Theorem~\ref{thm:fwd-cor-basic} and Section~\ref{sec:gluing}). We can state this formally because functions of real numbers have well-understood derivatives (Section~\ref{sec:dual-numbers-taylor}).
There are no established mathematical notions of derivatives at higher types, and
so we cannot even begin to argue that our syntactic derivatives of functions
$(\reals\to\reals)\to (\reals\to\reals)$ match with some existing mathematical notion (see also Section~\ref{derivatives-at-higher-types}). 

However, for functions of first-order type, like
$\List\reals\to\List\reals$, there \emph{are} established mathematical notions of derivative,
because we can understand
$\List\reals$ as the \emph{manifold} of all tuples of reals, and then appeal to the well-known theory of manifolds and jet bundles. We do this now to achieve a correctness theorem for all first-order types (Theorem~\ref{thm:fwd-cor-full}). The key high-level points are that
\begin{itemize}
\item manifolds support a notion of differentiation, and an interpretation of all first-order types, but not an interpretation of higher types;
\item diffeological spaces support all types, including higher types, but not an established notion of differentiation in general;
\item manifolds and smooth maps embed fully and faithfully in diffeological spaces, preserving the interpretation of first-order types, so we can use the two notions together.
\end{itemize}
We now explain this development in more detail.

For our purposes, a smooth manifold $M$ is a second-countable Hausdorff
topological space together with a smooth atlas. 
In more detail,
a topological space $X$ is second-countable when there exists a collection $U:=\{U_i\}_{i \in \mathbb{N}}$ of open subsets of $X$ 
such that any open subset of $X$ can be written as a union of elements of $U$. 
A topological space $X$ is Hausdorff if for every pair of distinct points $x$ and $y$, there exist disjoint open subsets $U,V$ of $X$ such that $x\in U, y\in V$.
A smooth atlas of a topological space $X$ is an open cover $\cover$ together
with homeomorphisms $\seq[U\in\cover]{\phi_U:U\to V_U\subseteq \RR^{n(U)}}$ onto open subsets (called charts or local coordinates) such that $\phi_U^{-1};\phi_V$ is smooth on its domain of definition for all $U,V\in\cover$.
A function $f:M\to N$ between manifolds is smooth if $\phi^{-1}_U;f;\psi_V$ is
smooth on its domain of definition for all charts $\phi_U$ and $\psi_V$ of $M$ and $N$, respectively.
Let us write $\Man$ for this category.
This definition of manifolds is a slight generalization of the more usual one from differential geometry because different charts in an atlas may have different finite dimensions $n(U)$. 
Thus we consider manifolds with dimensions that are potentially unbounded, albeit locally finite. 

Each open subset of $\RR^n$ can be regarded as a manifold. This lets us regard the category of manifolds $\Man$ as a full subcategory of the category of diffeological spaces. We consider a manifold $(X,\{\phi_U\}_U)$ as a diffeological space with the same carrier set $X$ whose plots $\plots X ^U$, called the \emph{manifold diffeology}, are the smooth functions in $\Man(U,X)$. A function $X\to Y$ is smooth in the sense of manifolds if and only if it is smooth in the sense of diffeological spaces~\cite{iglesias2013diffeology}. For the categorically minded reader, this means that we have a full embedding of $\Man$ into $\Diff$.
Moreover, the natural interpretation of the first-order fragment of our language in $\Man$ coincides with that in $\Diff$.
That is, the embedding of $\Man$ into $\Diff$ preserves finite products and countable coproducts (hence initial algebras of polynomial endofunctors).
\begin{prop}
  Suppose that a type $\ty$ is first-order, i.e.~it is just built from reals, products, variants, and lists (or, again, arbitrary inductive types), and not function types. Then the diffeological space $\sem{\ty}$ is a manifold. 
\end{prop}
\begin{proof}[Proof notes] This is proved by induction on the structure of types. In fact, one may show that every such $\sem\ty$ is isomorphic to a manifold of the form $\biguplus_{i=1}^n\RR^{d_i}$ where the bound $n$ is finite or $\infty$, but this isomorphism is typically not the identity map. 
\end{proof}

We recall how the Taylor representation of any morphism $f:M\to N$ of manifolds 
is given by its action on jets~\cite[Chapter IV]{kolar1999natural}.
For each point $x$ in a manifold $M$, define the $(k,R)$-\emph{jet space} $\Dsemsymbol[(k,R)]_x M$ to be the set $\{\gamma\in\Man(\RR^k,M)\mid \gamma(0)=x\}/\sim$ of equivalence classes $[\gamma]$ of $k$-dimensional plots $\gamma$
in $M$ based at $x$, where we identify $\gamma_1\sim \gamma_2$ iff all partial derivatives of order 
$\leq R$ coincide in the sense that

\[\frac{\partial^{\alpha_1+...+\alpha_k}(\gamma_1;f)}{\partial u_1^{\alpha_1}\cdots \partial u_k^{\alpha_k}} (0)
=\frac{\partial^{\alpha_1+...+\alpha_k}(\gamma_2;f)}{\partial u_1^{\alpha_1}\cdots \partial u_k^{\alpha_k}}(0)\] 

\noindent for all smooth $f:M\to\RR$ and all multi-indices $(\alpha_1,...,\alpha_k)=(0,...,0),...,(R,0,...,0)$.
In the case of $(k,R)=(1,1)$, a $(k,R)$-jet space is better known as a \emph{tangent space}.
The \emph{$(k,R)$-jet bundle} (a.k.a. \emph{tangent bundle}, in the case $(k,R)=(1,1)$) of $M$ is the set $\Dsem[(k,R)]{M}\defeq \biguplus_{x\in M} \Dsemsymbol[(k,R)]_x (M)$. The charts of $M$ equip $\Dsem[(k,R)]{M}$ with a canonical manifold structure.
The (manifold) diffeology of these jet bundles can be concisely summarized by the plots 
$\plots{\Dsem[(k,R)]{M}}^U=\set{f:U \to |\Dsem[(k,R)]{M}|\mid \exists g\in \plots{M}^{U\times \RR^k}. \forall u\in U. (g(u,0),[v\mapsto g(u,v)])=f(u)}$.\\
Then $\Dsemsymbol[(k,R)]$ acts on smooth maps $f:M\to N$ to give $\Dsem[(k,R)]{f}:\Dsem[(k,R)]{M}\to\Dsem[(k,R)]{N}$, defined as $\Dsem[(k,R)]{f}\sPair{x}{[\gamma]}\defeq \sPair{f(x)}{[\gamma;f]}$.
In local coordinates, this action $\Dsem[(k,R)]{f}$ is seen to coincide precisely with the 
$(k,R)$-Taylor representation of $f$ given by the Fa\`a di Bruno formula~\cite{merker2004four}.
Thus the $(k,R)$-jet bundle is a functor $\Dsemsymbol[(k,R)]:\Man\to\Man$~\cite{kolar1999natural}.

We can understand the jet bundle of a composite space in terms of that of its parts.
\begin{lem}\label{lemma:Dsem}
There are canonical isomorphisms $\Dsem[(k,R)]{\biguplus_{i=1}^\infty M_i} \cong \biguplus_{i=1}^\infty \Dsem[(k,R)]{M_i}$ and 
$\Dsem[(k,R)]{M_1\times\ldots \times M_n}\cong \Dsem[(k,R)]{M_1}\times\ldots\times \Dsem[(k,R)]{M_n}$.
\end{lem}
\begin{proof}[Proof notes]
For disjoint unions, notice that smooth morphisms from $\RR^k$ into a 
disjoint union of manifolds always factor through a single inclusion, because $\RR^k$ is connected.
For products, it is well-known that partial derivatives of a morphism $\sTuple{f_1,...,f_n}$ are calculated 
component-wise~\cite[ex. 3-2]{lee2013smooth}.
\end{proof}

We define a canonical isomorphism $\DtoT{\ty}:\sem{\Dsyn[(k,R)]{\ty}}\to\Dsem[(k,R)]{\sem{\ty}}$ for every first-order type $\ty$, by induction on the structure of types. We let $\DtoT{\reals}:\sem{\Dsyn[(k,R)]{\reals}}\to\Dsem[(k,R)]{\sem{\reals}}$ send a tuple $(x_\alpha)_{|\alpha|\leq R}$ to
$\big(x_{0,\ldots,0},[t\mapsto \sum_{|\alpha|\leq R}\frac{x_\alpha}{\alpha_1!\cdots\alpha_k!}t_1^{\alpha_1}\cdots t_k^{\alpha_k}]\big)$.
For the other types, we use Lemma~\ref{lemma:Dsem}. 
We can now phrase correctness at all first-order types.
 \begin{thm}[Semantic correctness of ${\Dsynsymbol[(k,R)]}$ (full)]\label{thm:fwd-cor-full}
   For any first-order type $\ty$, any first-order context $\Gamma$,
   and any term $\Gamma\vdash\trm:\ty$,
   the syntactic translation $\Dsynsymbol[(k,R)]$ coincides with the $(k,R)$-jet bundle functor, modulo these canonical isomorphisms:
   \[\xymatrix{
       \sem{\Dsyn[(k,R)]\Gamma}\ar[rr]^{\sem{\Dsyn[(k,R)] \trm}}\ar[d]_{\DtoT{\Gamma}}^\cong
       &&\sem{\Dsyn[(k,R)] \ty}\ar[d]^{\DtoT{\ty}}_\cong
       \\
       \Dsem[(k,R)]{\sem\Gamma}\ar[rr]_{\Dsem[(k,R)]{\sem \trm}}
       &&\Dsem[(k,R)] {\sem \ty}
       }\]
 \end{thm}

 \begin{proof}[Proof notes]
   For any $k$-dimensional plot $\gamma\in\Man(\RR^k,M)$, let $\bar \gamma\in\Man(\RR^k,\Dsem[(k,R)] M)$ be the $(k,R)$-jet curve, given by $\bar \gamma(x)=(\gamma(x),[t\mapsto \gamma(x+t)])$.
First, we note that a smooth map $h:\Dsem[(k,R)] M\to \Dsem[(k,R)] N$ is of the form $\Dsem[(k,R)]{g}$ for some $g:M\to N$ if for all smooth $\gamma:\RR^k\to M$ we have $\bar \gamma;h=\overline{(\gamma;g)}:\RR^k\to \Dsem[(k,R)] N$. This generalizes~(\ref{eqn:taylorrepresentation}).
Second, for any first-order type $\ty$, $S_{\ty}=\{(f,\tilde{f})~|~\tilde{f};\DtoT{\ty}=\bar f\}$. This is shown by induction on the structure of types. 
We conclude the theorem from diagram (\ref{dgm:gluing}), by combining these two observations.

 \end{proof}

\section{Discussion: What are derivatives of higher-order functions?}\label{derivatives-at-higher-types}
In our gluing categories $\Gl[k]$ of Section~\ref{sec:gluing}, we have avoided the question of what semantic derivatives should be associated with higher-order functions. 
Our syntactic macro $\Dsynsymbol$ provides a specific derivative for every definable function, 
but in the model $\Gl[k]$ there is only a \emph{relation} between plots and their corresponding Taylor representations, and this relation is not necessarily single-valued. 
Our approach has been deliberately neutral about what ``the'' correct derivative of a higher-order function should be.
Instead, what matters is that we are using ``a'' derivative that is correct in the sense that it can never be used to produce incorrect derivatives for first-order functions, where we do have an unambiguous notion of correct derivative.

\subsection{Automatic derivatives of higher-order functions may not be unique!}
For a concrete example showing that derivatives of higher-order functions might not be unique in our framework, let us consider the case $(k,R)=(1,1)$ and focus on first derivatives of the evaluation 
function
\begin{align*}
\ev:\ &\RR\to \sem{(\reals\To\reals)\To \reals}=\RR\To (\RR\Rightarrow \RR)\Rightarrow \RR;\\
&r\mapsto(f\mapsto f(r)).
\end{align*}
Our macro $\Dsynsymbol$ will return $\lambda a:\RR. \lambda f:\RR\times\RR\Rightarrow \RR\times\RR. f(a,1)$.
In this section we show that the $\lambda$-term $\lambda a:\RR. \lambda f:\RR\times\RR\Rightarrow \RR\times\RR. \sort f(a,1)$ 
is also a valid derivative of the evaluation map, where $\sort: (\RR\times\RR\Rightarrow \RR\times\RR)\Rightarrow (\RR\times\RR\Rightarrow \RR\times\RR)$ is defined by
\begin{align*}
    \sort:=\lambda f.\lambda (r,r').(\pi_1(f(r,0)),\pi_2(f(r,0))+r'\cdot\nabla((-,0);f;\pi_1)(r)).
\end{align*}
This map is idempotent and it converts any map $\RR \times \RR \To \RR \times \RR$ into a map that is affine in the tangent component and has the correct slope there.
For example, $(\sort(\swap))(r,r')=(0,r)$, where we write
\begin{align*}
    \swap:\ &\RR\times\RR\to \RR\times\RR\\
    &(r,r')\mapsto (r',r).
    \end{align*}

According to our gluing semantics, a function $g:\semgl{\ty}_1\to \semgl{\ty[2]}_1$ defines \emph{a} correct $(k,R)$-Taylor representation of a function $f:\semgl{\ty}_0\to \semgl{\ty[2]}_0$ iff $(f,g)$ defines a morphism 
$\semgl{\ty}\to\semgl{\ty[2]}$ in $\Gl[k]$.
In particular, there is no guarantee that every $f$ has a \emph{unique} correct $(k,R)$-Taylor representation $g$.
(Although such Taylor representations are, in fact, unique when $\ty,\ty[2]$ are first-order types.)
The gluing relation $\semgl{(\reals\To\reals)\To \reals}$ in $\Gl[1]$ relates curves $\gamma:\RR\To (\RR\Rightarrow \RR)\Rightarrow \RR$ to ``tangent curves'' $\gamma': \RR\To (\RR\times \RR\Rightarrow \RR\times \RR)\Rightarrow \RR\times\RR$.
In this relation, the function $\ev$ is related to at least two different tangent curves.

\begin{lem}
We have a smooth map 
\begin{align*}
\sort:\ &(\RR\times\RR\Rightarrow \RR\times\RR)\to (\RR\times\RR\Rightarrow \RR\times\RR)\\
&f\mapsto ((r,r')\mapsto (\pi_1(f(r,0)),\pi_2(f(r,0))+r'\cdot\nabla ((-,0);f;\pi_1)(r))).
\end{align*}
\end{lem}
\begin{proof}
Let $f\in \plots{\RR\times\RR\Rightarrow \RR\times\RR}^U$ and let $\gamma_1,\gamma_2\in \plots{\RR}^U$.
By definition of the exponential in $\Diff$, the map $F:U\times\RR\to\RR\times\RR$ given by
$F(u,r)=f(u)(r,0)$ is smooth. Hence both $\pi_1 F$ and $\pi_2 F$ are smooth, and so is the partial derivative
$\frac{\partial(\pi_1F)}{\partial r}:U\times\RR\to\RR$.
Consequently,
\[u\mapsto (f;\sort)(u)(\gamma_1(u),\gamma_2(u))=
\left(\pi_1F(u,\gamma_1(u)),\pi_2F(u,\gamma_1(u))+\gamma_2(u)\cdot\frac{\partial(\pi_1F)}{\partial r}(u,\gamma_1(u))\right)
\]
is in $\plots{\RR\times\RR}^U$, by definition of the product in $\Diff$.
It follows that $(f;\sort) \in \plots{\RR\times\RR\Rightarrow \RR\times\RR}^U$.
\end{proof}

\begin{prop}
We have both $(\ev,\ev'_1)\in \semgl{(\reals\To\reals)\To \reals}$ and 
$(\ev,\ev'_2)\in \semgl{(\reals\To\reals)\To \reals}$ for 
\begin{align*}
\ev'_1 : & \RR \to (\RR\times\RR\Rightarrow \RR\times\RR)\Rightarrow \RR\times\RR\\
&a\mapsto (f\mapsto f(a,1))\\
\ev'_2 : & \RR \to (\RR\times\RR\Rightarrow \RR\times\RR)\Rightarrow \RR\times\RR\\
&a\mapsto (f\mapsto (\sort f)(a,1)).
\end{align*}
\end{prop}
\begin{proof}
    By definition of $\semgl{-}$, we need to show that for any $(\gamma,\gamma')\in \semgl{\reals\To\reals}$, we have
$(x\mapsto \ev(x)(\gamma(x)), x\mapsto \ev'_i(x)(\gamma'(x)))\in\semgl{\reals}$.
This means that we need to show that for $i=1,2$
\begin{align*}
x\mapsto \ev'_i(x)(\gamma'(x))=(x\mapsto \ev(x)(\gamma(x)),\nabla(x\mapsto \ev(x)(\gamma(x))))
\end{align*}
Unrolling further, this means we need to show that for any $\gamma:\RR\to\RR \Rightarrow\RR$ and $\gamma':\RR\to\RR\times\RR \Rightarrow\RR\times\RR$ such that for any $(\delta,\delta')\in\semgl{\reals}$ (which means that 
$\delta:\RR\to\RR$ and
$\delta'=(\delta,\nabla \delta)$), we have 
\begin{align*}
\Big(r\mapsto \gamma(r)(\delta(r)),r\mapsto \gamma'(r)(\delta'(r))\Big)\in
\semgl{\reals}
\end{align*}
The latter condition means that we need to show that 
\begin{align*}
r\mapsto \gamma'(r)(\delta(r),\nabla \delta(r))=(r\mapsto \gamma(r)(\delta(r)), \nabla(r\mapsto \gamma(r)(\delta(r))))
\end{align*}

For $\ev'_1$, we need to show that
\begin{align*}
x\mapsto \ev'_1(x)(\gamma'(x))=(x\mapsto \ev(x)(\gamma(x)),\\\nabla(x\mapsto \ev(x)(\gamma(x))))
\end{align*}
After inlining the definition of $\ev'_1$, we need to show that
\begin{align*}
x\mapsto \gamma'(x)(x,1)=(x\mapsto \gamma(x)(x),\nabla(x\mapsto \gamma(x)(x)))
\end{align*}
This follows from the assumption by choosing $\delta(r)=r$, and hence $\delta'(r)=(r,1)$.

For $\ev'_2$, after inlining the definition of $\sort$, we need to show that
\begin{align*}
x\mapsto (&\pi_1(\gamma'(x)(x,0)),\pi_2(\gamma'(x)(x,0))+\nabla((-,0);\gamma'(x);\pi_1)(x))\\
&=(x\mapsto \gamma(x)(x),\nabla(x\mapsto \gamma(x)(x))).
\end{align*}
The first components agree because the assumption $(\gamma,\gamma')\in \semgl{\reals\To\reals}$, applied to the constant curve $\delta_a(x)=a$, gives
$\pi_1(\gamma'(x)(a,0))=\gamma(x)(a)$ for all $a$ and $x$.
For the second components, the same constant-curve argument gives
$\pi_2(\gamma'(x)(a,0))=\frac{\partial}{\partial x}\gamma(x)(a)$, while differentiating
$\pi_1(\gamma'(x)(a,0))=\gamma(x)(a)$ with respect to $a$ gives
$\nabla((-,0);\gamma'(x);\pi_1)(x)=\frac{\partial}{\partial a}\gamma(x)(a)|_{a=x}$.
The chain rule gives the desired derivative of $x\mapsto \gamma(x)(x)$.
\end{proof}
However, $\ev'_1\neq \ev'_2$ since $\ev'_1(a)(\swap)=(1,a)$ and $\ev'_2(a)(\swap)=(0,a)$.
This shows that $\ev'_1$ and $\ev'_2$ are both ``valid'' semantic derivatives of the evaluation function $(\ev)$ 
in our framework.
In particular, it shows that semantic derivatives of higher-order functions might not be unique. 
Our macro $\Dsynsymbol$ will return $\ev'_1$, but everything would still work just as well if it instead returned $\ev'_2$. 

\subsection{Canonical derivatives of higher-order functions?}
Differential geometers and analysts have long pursued notions of a canonical derivative of various 
higher-order functions arising, for example, in the calculus of variations and in the study of 
infinite-dimensional Lie groups~\cite{kriegl1997convenient}.
Uncontroversial notions of derivative exist on various infinite-dimensional spaces of functions
that form suitable (so-called convenient) vector spaces, or manifolds locally modelled on such vector spaces.
At the level of generality of diffeological spaces, however, various natural notions of derivative that 
coincide in convenient vector spaces start to diverge, and it is no longer clear what the best definition 
of a derivative is~\cite{christensen2014tangent}.
Another, fundamentally different setting that defines canonical derivatives of many
higher-order functions is given by 
synthetic differential geometry~\cite{kock2006synthetic}.

While derivatives of higher-order functions are of deep interest and have rightly been studied 
in their own right in differential geometry, 
we believe the situation is subtly different in computer science:
\begin{enumerate}
\item In programming applications, we use higher-order programs only to construct the first-order 
functions that we ultimately run and differentiate.
Automatic differentiation methods can exploit this freedom: derivatives of higher-order functions only matter insofar as they can be used to construct the correct derivatives of first-order functions,
so we can choose a simple and inexpensive notion of derivative among the valid options.
Thus, the fact that our semantics does not commit to a single notion of derivative of higher-order 
functions can be seen as a \emph{feature rather than a bug} that models the pragmatics of programming practice.
\item While function spaces in differential geometry are typically infinite-dimensional objects that are 
unsuitable for representation in the finite memory of a computer,
higher-order functions as used in programming are much more restricted: all they can do is call 
a function on finitely many arguments and analyse the function outputs.
Thus, function types in programming can be thought of as (locally) finite-dimensional.
If a canonical notion of automatic derivative of a higher-order function is desired,
it may be worth pursuing a more intensional notion of semantics such as one based on game semantics.
Such intensional techniques could capture the computational notion of a higher-order function
better than our current (and other) extensional semantics using existing techniques from differential geometry.
We hope that an exploration of such techniques might lead to an appropriate notion of computable derivative, even for higher-order functions.
\end{enumerate}

\section{Discussion and future work}\label{sec:conclusion}

\subsection{Summary}
We have shown that diffeological spaces provide a denotational semantics for a higher-order language with variants and inductive types (Sections~\ref{sec:semantics} and~\ref{sec:extended-language}). We have used this to show the correctness of simple forward-mode AD translations for calculating higher derivatives (Theorem~\ref{thm:fwd-cor-basic}, Theorem~\ref{thm:fwd-cor-full}).

The structure of our elementary correctness argument for Theorem~\ref{thm:fwd-cor-basic} is a typical logical relations proof over a denotational semantics.
As explained in Section~\ref{sec:correctness}, this can equivalently be understood as a denotational semantics in a new kind of space obtained by categorical gluing.

Overall, then, there are two logical relations at play. One is in diffeological spaces, which ensures that all definable functions are smooth. The other is in the correctness proof (equivalently in the categorical gluing), which explicitly tracks both the derivative of each function and the syntactic AD even at higher types. 

\subsection{Connection to the state of the art in AD implementation}
As is common in denotational semantics research, we have focused here on an idealized language and simple translations to illustrate the main aspects of the method. There are a number of points where our approach is simplistic compared to the current state of the art, as we now explain.

\subsubsection{Representation of vectors}
\label{sub:discussion}
In our examples we have treated $n$-vectors as tuples of length~$n$. This style of programming does not scale to large~$n$. A better solution would be to use array types, following~\cite{shaikhha2019efficient}.
As demonstrated by~\cite{chinjensem2020formalized}, our categorical semantics and correctness proofs extend straightforwardly to cover them, in a way analogous to our treatment of lists.
In fact, that work formalizes our correctness arguments in Coq and extends them to apply 
to the system of~\cite{shaikhha2019efficient}.

\subsubsection{Efficient forward-mode AD}
For AD to be useful, it must be fast.
The $(1,1)$-AD macro $\Dsynsymbol[(1,1)]$ that we use is the basis of an efficient AD library~\cite{shaikhha2019efficient}. Numerous optimizations are needed, ranging from algebraic manipulations to partial evaluation and the use of an optimizing C compiler, but the resulting implementation performs well in experiments~\cite{shaikhha2019efficient}. The Coq formalization~\cite{chinjensem2020formalized} validates some of these manipulations using a semantics similar to ours. 
We believe the implementation in~\cite{shaikhha2019efficient} can be extended to apply to the more general 
$(k,R)$-AD methods we described in this paper with minor changes.

\subsubsection{Reverse-mode and mixed-mode AD}
While forward-mode AD methods are useful, many applications require reverse-mode AD,
or even mixed-mode AD for efficiency.
In~\cite{huot2020correctness}, we described how our correctness proof applies to a continuation-based 
AD technique that closely resembles reverse-mode AD, but only has the correct complexity under a 
non-standard operational semantics~\cite{brunel2019backpropagation} (in particular, the linear factoring rule is crucial).
It remains to be seen whether this technique and its correctness proof can be adapted to yield 
genuine reverse-mode AD under a standard operational semantics.

Alternatively, by relying on a variation of our techniques, V\'ak\'ar~\cite{vakar2021reverse} gives a correctness proof of a rather different $(1,1)$ reverse-mode AD algorithm that stores the (primal, adjoint)-vector pair 
as a struct-of-arrays rather than as an array-of-structs.
Future work could explore extending its analysis to $(k,R)$ reverse-mode AD and mixed-mode AD for efficiently computing higher-order derivatives.

\subsubsection{Other language features}
\label{sub:summary_and_future_work}
The idealized languages considered here do not cover several useful language constructs, including partial functions (such as division), partially smooth functions (such as ReLU), phenomena such as iteration and recursion, and probabilities.  
Recent work by V\'ak\'ar~\cite{vakar2020denotational} shows how our analysis of $(1,1)$-AD extends to apply to partiality, iteration, and recursion. 
This development is orthogonal to the one in this paper: its methods combine directly with those in the 
present paper to analyze $(k,R)$-forward-mode AD of recursive programs.
We leave the analysis of AD for probabilistic programs for future work.

\section*{Acknowledgments}
  \noindent We have benefited from discussing this work with many people, including M.~Betancourt, B. Carpenter, O.~Kammar, C.~Mak, L.~Ong, B.~Pearlmutter, G.~Plotkin, A.~Shaikhha, J.~Sigal, and others. 
  In the course of this work, MV was also employed at Oxford (EPSRC Project EP/M023974/1) and at Columbia on the Stan development team.
  This project has also received funding from the European Union’s Horizon 2020 research and innovation 
  programme under the Marie Skłodowska-Curie grant agreement No. 895827;
  NWO Veni grant number VI.Veni.202.124;
  a Royal Society University Research Fellowship; the ERC BLAST grant;
  the Air Force Office of Scientific Research under award number FA9550-21-1-0038; and a Facebook Research Award.

\clearpage
\bibliographystyle{alphaurl}
\bibliography{bibliography}

@article{baez2011convenient,
  title={Convenient categories of smooth spaces},
  author={Baez, John and Hoffnung, Alexander},
  journal={Transactions of the American Mathematical Society},
  volume={363},
  number={11},
  pages={5789--5825},
  year={2011}
}

@book{iglesias2013diffeology,
  title={Diffeology},
  author={Iglesias-Zemmour, Patrick},
  OPTvolume={185},
  year={2013},
  publisher={American Mathematical Soc.}
}

@article{christensen2014tangent,
  title={Tangent spaces and tangent bundles for diffeological spaces},
  author={Christensen, J Daniel and Wu, Enxin},
  journal={arXiv preprint arXiv:1411.5425},
  year={2014}
}

@article{wang2018demystifying,
  title={Demystifying differentiable programming: Shift/reset the penultimate backpropagator},
  author={Wang, Fei and Wu, Xilun and Essertel, Gregory and Decker, James and Rompf, Tiark},
  journal={Proceedings of the ACM on Programming Languages},
  volume={3},
  number={ICFP},
  OPTpages={97},
  year={2019},
  publisher={ACM}
}

@article{betancourt2018geometric,
  title={A geometric theory of higher-order automatic differentiation},
  author={Betancourt, Michael},
  journal={arXiv preprint arXiv:1812.11592},
  year={2018}
}

@article{elliott2018simple,
  title={The simple essence of automatic differentiation},
  author={Elliott, Conal},
  journal={Proceedings of the ACM on Programming Languages},
  volume={2},
  number={ICFP},
  pages={70},
  year={2018},
  publisher={ACM}
}

@article{pearlmutter2008reverse,
  title={Reverse-mode {AD} in a functional framework: Lambda the ultimate backpropagator},
  author={Pearlmutter, Barak A and Siskind, Jeffrey Mark},
  journal={ACM Transactions on Programming Languages and Systems (TOPLAS)},
  volume={30},
  number={2},
  pages={7},
  year={2008},
  publisher={ACM}
}

@inproceedings{abadi2016tensorflow,
  title={Tensorflow: A system for large-scale machine learning},
author={
    Mart\'{i}n~Abadi and
    Ashish~Agarwal and
    Paul~Barham and
    Eugene~Brevdo and
    Zhifeng~Chen and
    Craig~Citro and
    Greg~S.~Corrado and
    Andy~Davis and
    Jeffrey~Dean and
    Matthieu~Devin and
    Sanjay~Ghemawat and
    Ian~Goodfellow and
    Andrew~Harp and
    Geoffrey~Irving and
    Michael~Isard and
    Yangqing Jia and
    Rafal~Jozefowicz and
    Lukasz~Kaiser and
    Manjunath~Kudlur and
    Josh~Levenberg and
    Dandelion~Man\'{e} and
    Rajat~Monga and
    Sherry~Moore and
    Derek~Murray and
    Chris~Olah and
    Mike~Schuster and
    Jonathon~Shlens and
    Benoit~Steiner and
    Ilya~Sutskever and
    Kunal~Talwar and
    Paul~Tucker and
    Vincent~Vanhoucke and
    Vijay~Vasudevan and
    Fernanda~Vi\'{e}gas and
    Oriol~Vinyals and
    Pete~Warden and
    Martin~Wattenberg and
    Martin~Wicke and
    Yuan~Yu and
    Xiaoqiang~Zheng},
  booktitle={12th {U}{S}{E}{N}{I}{X} Symposium on Operating Systems Design and Implementation ({O}{S}{D}{I} 16)},
  pages={265--283},
  year=2016
}

@article{carpenter2015stan,
  title={The {S}tan math library: Reverse-mode automatic differentiation in {C}++},
  author={Carpenter, Bob and Hoffman, Matthew D and Brubaker, Marcus and Lee, Daniel and Li, Peter and Betancourt, Michael},
  journal={arXiv preprint arXiv:1509.07164},
  year={2015}
}

@article{qian1999momentum,
  title={On the momentum term in gradient descent learning algorithms},
  author={Qian, Ning},
  journal={Neural networks},
  volume={12},
  number={1},
  pages={145--151},
  year={1999},
  publisher={Elsevier}
}

@article{kingma2014adam,
  title={Adam: A method for stochastic optimization},
  author={Kingma, Diederik P and Ba, Jimmy},
  journal={arXiv preprint arXiv:1412.6980},
  year={2014}
}

@article{duchi2011adaptive,
  title={Adaptive subgradient methods for online learning and stochastic optimization},
  author={Duchi, John and Hazan, Elad and Singer, Yoram},
  journal={Journal of Machine Learning Research},
  volume={12},
  number={Jul},
  pages={2121--2159},
  year={2011}
}

@article{liu1989limited,
  title={On the limited memory {B}{F}{G}{S} method for large scale optimization},
  author={Liu, Dong C and Nocedal, Jorge},
  journal={Mathematical programming},
  volume={45},
  number={1-3},
  pages={503--528},
  year={1989},
  publisher={Springer}
}

@article{robbins1951stochastic,
  title={A stochastic approximation method},
  author={Robbins, Herbert and Monro, Sutton},
  journal={The Annals of Mathematical Statistics},
  pages={400--407},
  year={1951},
  publisher={JSTOR}
}

@article{kiefer1952stochastic,
  title={Stochastic estimation of the maximum of a regression function},
  author={Kiefer, Jack and Wolfowitz, Jacob and others},
  journal={The Annals of Mathematical Statistics},
  volume={23},
  number={3},
  pages={462--466},
  year={1952},
  publisher={Institute of Mathematical Statistics}
}

@article{shaikhha2019efficient,
  title={Efficient differentiable programming in a functional array-processing language},
  author={Shaikhha, Amir and Fitzgibbon, Andrew and Vytiniotis, Dimitrios and Peyton Jones, Simon},
  journal={Proceedings of the ACM on Programming Languages},
  volume=3,
  number={ICFP},
  pages=97,
  year=2019,
  publisher={ACM}
}

@inproceedings{fong2019backprop,
  title={Backprop as functor: A compositional perspective on supervised learning},
  author={Fong, Brendan and Spivak, David and Tuy{\'e}ras, R{\'e}my},
  booktitle={2019 34th Annual ACM/IEEE Symposium on Logic in Computer Science (LICS)},
  pages={1--13},
  year={2019},
  organization={IEEE}
}

@Unpublished{plotkin-invited-talk,
  author = 	 {Gordon D Plotkin},
  title = 	 {Some Principles of Differential Programming Languages},
  note = 	 {Invited talk, POPL 2018},
  OPTkey = 	 {},
  OPTmonth = 	 {},
  year = 	 {2018},
  OPTannote = 	 {}
}

@article{ehrhard2003differential,
  title={The differential lambda-calculus},
  author={Ehrhard, Thomas and Regnier, Laurent},
  journal={Theoretical Computer Science},
  volume={309},
  number={1-3},
  pages={1--41},
  year={2003},
  publisher={Elsevier}
}

@book{kock2006synthetic,
  title={Synthetic differential geometry},
  author={Kock, Anders},
  volume={333},
  year={2006},
  publisher={Cambridge University Press}
}

@incollection{souriau1980groupes,
  title={Groupes diff{\'e}rentiels},
  author={Souriau, Jean-Marie},
  booktitle={Differential geometrical methods in mathematical physics},
  pages={91--128},
  year={1980},
  publisher={Springer}
}

@article{hoffman2014no,
  title={The {N}o-{U}-{T}urn sampler: adaptively setting path lengths in {H}amiltonian {M}onte {C}arlo.},
  author={Hoffman, Matthew D and Gelman, Andrew},
  journal={Journal of Machine Learning Research},
  volume={15},
  number={1},
  pages={1593--1623},
  year={2014}
}

@InCollection{neal2011mcmc,
  title={M{C}{M}{C} using {H}amiltonian dynamics},
  author={Neal, Radford M},
  booktitle={Handbook of {M}arkov {C}hain {M}onte {C}arlo},
  chapter={5},
  year={2011},
  publisher={Chapman \& Hall / CRC Press}
}

@article{kucukelbir2017automatic,
  title={Automatic differentiation variational inference},
  author={Kucukelbir, Alp and Tran, Dustin and Ranganath, Rajesh and Gelman, Andrew and Blei, David M},
  journal={The Journal of Machine Learning Research},
  volume={18},
  number={1},
  pages={430--474},
  year={2017},
  publisher={JMLR. org}
}

@inproceedings{brunel2019backpropagation,
  title={Backpropagation in the Simply Typed Lambda-calculus with Linear Negation},
  author={Brunel, Alois and Mazza, Damiano and Pagani, Michele},
  booktitle={Proc.~POPL 2020},
  year={2020}
}

@inproceedings{mitchell1992notes,
  title={Notes on sconing and relators},
  author={Mitchell, John C and Scedrov, Andre},
  booktitle={International Workshop on Computer Science Logic},
  pages={352--378},
  year={1992},
  organization={Springer}
}

@techreport{pitts1995categorical,
  title={Categorical logic},
  author={Pitts, Andrew M},
  year={1995},
  institution={University of Cambridge, Computer Laboratory}
}

@InProceedings{johnstone-lack-sobocinski,
  author = 	 {Peter T Johnstone and Stephen Lack and P Sobocinski},
  title = 	 {Quasitoposes, Quasiadhesive Categories and {A}rtin Glueing},
  OPTcrossref =  {},
  OPTkey = 	 {},
  booktitle = {Proc.~CALCO 2007},
  year = 	 {2007},
  OPTeditor = 	 {},
  OPTvolume = 	 {},
  OPTnumber = 	 {},
  OPTseries = 	 {},
  OPTpages = 	 {},
  OPTmonth = 	 {},
  OPTaddress = 	 {},
  OPTorganization = {},
  OPTpublisher = {},
  OPTnote = 	 {},
  OPTannote = 	 {}
}

@incollection{lee2013smooth,
  title={Smooth manifolds},
  author={Lee, John M},
  booktitle={Introduction to Smooth Manifolds},
  pages={1--31},
  year={2013},
  publisher={Springer}
}

@Article{smootheology,
  author = 	 {Andrew Stacey},
  title = 	 {Comparative smootheology},
  journal = 	 {Theory Appl. Categ.},
  year = 	 {2011},
  OPTkey = 	 {},
  volume = 	 {25},
  number = 	 {4},
  pages = 	 {64--117},
  OPTmonth = 	 {},
  OPTnote = 	 {},
  OPTannote = 	 {}
}

@InProceedings{gallagher-sdg,
  author = 	 {Geoff Cruttwell and Jonathan Gallagher and Ben MacAdam},
  title = 	 {Towards formalizing and extending differential programming using tangent categories},
  OPTcrossref =  {},
  OPTkey = 	 {},
  booktitle = {Proc.~ACT 2019},
  year = 	 {2019},
  OPTeditor = 	 {},
  OPTvolume = 	 {},
  OPTnumber = 	 {},
  OPTseries = 	 {},
  OPTpages = 	 {},
  OPTmonth = 	 {},
  OPTaddress = 	 {},
  OPTorganization = {},
  OPTpublisher = {},
  OPTnote = 	 {},
  OPTannote = 	 {}
}

@InProceedings{bcdg-open-logical-relations,
  author = 	 {Gilles Barthe and Rapha\"elle Crubill\'e and Ugo Dal~Lago and Francesco Gavazzo},
  title = 	 {On the Versatility of Open Logical Relations:
Continuity, Automatic Differentiation, and a Containment Theorem},
  OPTcrossref =  {},
  OPTkey = 	 {},
  booktitle = {Proc.~ESOP 2020},
  year = 	 {2020},
  OPTeditor = 	 {},
  OPTvolume = 	 {},
  OPTnumber = 	 {},
  OPTseries = 	 {},
  OPTpages = 	 {},
  OPTmonth = 	 {},
  OPTaddress = 	 {},
  OPTorganization = {},
  publisher = {Springer},
  note = 	 {To appear},
  OPTannote = 	 {}
}

@InProceedings{abadi-plotkin2020,
  author = 	 {Mart\'in Abadi and Gordon D Plotkin},
  title = 	 {A Simple Differentiable Programming Language},
  OPTcrossref =  {},
  OPTkey = 	 {},
  booktitle = {Proc.~POPL 2020},
  year = 	 {2020},
  OPTeditor = 	 {},
  OPTvolume = 	 {},
  OPTnumber = 	 {},
  OPTseries = 	 {},
  OPTpages = 	 {},
  OPTmonth = 	 {},
  OPTaddress = 	 {},
  OPTorganization = {},
  publisher = {ACM},
  OPTnote = 	 {},
  OPTannote = 	 {}
}

@InProceedings{rev-deriv-cat2020,
  author = 	 {J. Robin B. Cockett and Geoff S. H. Cruttwell and Jonathan Gallagher and Jean-Simon Pacaud Lemay and Benjamin MacAdam and Gordon D. Plotkin and Dorette Pronk},
  title = 	 {Reverse Derivative Categories},
  OPTcrossref =  {},
  OPTkey = 	 {},
  booktitle = {Proc.~CSL 2020},
  year = 	 {2020},
  OPTeditor = 	 {},
  OPTvolume = 	 {},
  OPTnumber = 	 {},
  OPTseries = 	 {},
  OPTpages = 	 {},
  OPTmonth = 	 {},
  OPTaddress = 	 {},
  OPTorganization = {},
  OPTpublisher = {},
  OPTnote = 	 {},
  OPTannote = 	 {}
}

@InProceedings{Manzyuk2012,
  author = 	 {Oleksandr Manzyuk},
  title = 	 {A Simply Typed $\lambda$-Calculus of Forward Automatic Differentiation},
  OPTcrossref =  {},
  OPTkey = 	 {},
  booktitle = {Proc.~MFPS 2012},
  year = 	 {2012},
  OPTeditor = 	 {},
  OPTvolume = 	 {},
  OPTnumber = 	 {},
  OPTseries = 	 {},
  OPTpages = 	 {},
  OPTmonth = 	 {},
  OPTaddress = 	 {},
  OPTorganization = {},
  OPTpublisher = {},
  OPTnote = 	 {},
  OPTannote = 	 {}
}

@Unpublished{mak-ong2020,
  author = 	 {Carol Mak and Luke Ong},
  title = 	 {A Differential-form Pullback Programming Language for Higher-order Reverse-mode Automatic Differentiation},
  note = 	 {arxiv:2002.08241},
  OPTkey = 	 {},
  OPTmonth = 	 {},
  year = 	 {2020},
  OPTannote = 	 {}
}

@Misc{hsv-fossacs2020,
  author = 	 {Mathieu Huot and Sam Staton and Matthijs Vákár},
  title = 	 {Correctness of Automatic Differentiation via Diffeologies and Categorical Gluing},
  howpublished = {Full version},
  OPTmonth = 	 {},
  year = 	 {2020},
  note = 	 {arxiv:2001.02209},
  OPTannote = 	 {}
}

@article{kolar1999natural,
  title={Natural operations in differential geometry},
  author={Kol{\'a}r, Ivan and Slov{\'a}k, Jan and Michor, Peter W},
  year={1999}
}

@book{kriegl1997convenient,
  title={The convenient setting of global analysis},
  author={Kriegl, Andreas and Michor, Peter W},
  volume={53},
  year={1997},
  publisher={American Mathematical Soc.}
}

@article{griewank2000evaluating,
  title={Evaluating higher derivative tensors by forward propagation of univariate Taylor series},
  author={Griewank, Andreas and Utke, Jean and Walther, Andrea},
  journal={Mathematics of Computation},
  volume={69},
  number={231},
  pages={1117--1130},
  year={2000}
}

@inproceedings{huot2020correctness,
  title={Correctness of Automatic Differentiation via Diffeologies and Categorical Gluing.},
  author={Huot, Mathieu and Staton, Sam and V{\'a}k{\'a}r, Matthijs},
  booktitle={FoSSaCS},
  pages={319--338},
  year={2020}
}

@article{merker2004four,
  title={Four explicit formulas for the prolongations of an infinitesimal Lie symmetry and multivariate {F}aa di {B}runo formulas},
  author={Merker, Joel},
  journal={arXiv preprint math/0411650},
  year={2004}
}

@article{constantine1996multivariate,
  title={A multivariate {F}aa di {B}runo formula with applications},
  author={Constantine, G and Savits, T},
  journal={Transactions of the American Mathematical Society},
  volume={348},
  number={2},
  pages={503--520},
  year={1996}
}

@article{encinas2003short,
  title={A short proof of the generalized {F}a{\`a} di {B}runo's formula},
  author={Encinas, L Hern{\'a}ndez and Masque, J Munoz},
  journal={Applied Mathematics Letters},
  volume={16},
  number={6},
  pages={975--979},
  year={2003},
  publisher={Elsevier}
}

@article{savits2006some,
  title={Some statistical applications of {F}aa di {B}runo},
  author={Savits, Thomas H},
  journal={Journal of Multivariate Analysis},
  volume={97},
  number={10},
  pages={2131--2140},
  year={2006},
  publisher={Elsevier}
}

@article{chinjensem2020formalized,
  title={Formalized Correctness Proofs of Automatic Differentiation in {C}oq},
  author={Chin Jen Sem, Curtis},
  journal={Master's Thesis, Utrecht University},
  note = {Thesis: https://dspace.library.uu.nl/handle/1874/400790. Coq code: https://github.com/crtschin/thesis},
  year={2020}
}

@inproceedings{van2018automatic,
  title={Automatic differentiation in {M}{L}: Where we are and where we should be going},
  author={Van Merri{\"e}nboer, Bart and Breuleux, Olivier and Bergeron, Arnaud and Lamblin, Pascal},
  booktitle={Advances in Neural Information Processing Systems},
  pages={8757--8767},
  year={2018}
}

@inproceedings{martens2010deep,
  title={Deep learning via {H}essian-free optimization.},
  author={Martens, James},
  booktitle={ICML},
  volume={27},
  pages={735--742},
  year={2010}
}

@inproceedings{lee2020correctness,
  title={On Correctness of Automatic Differentiation for Non-Differentiable Functions},
  author={Lee, Wonyeol and Yu, Hangyeol and Rival, Xavier and Yang, Hongseok},
  booktitle={Advances in Neural Information Processing Systems},
  year={2020}
}

@article{sherman2020lambda_s,
  title={{$\lambda_S $}: Computable semantics for differentiable programming with higher-order functions and datatypes},
  author={Sherman, Benjamin and Michel, Jesse and Carbin, Michael},
  journal={arXiv preprint arXiv:2007.08017},
  year={2020}
}

@article{bernstein2020differentiating,
  title={Differentiating a Tensor Language},
  author={Bernstein, Gilbert and Mara, Michael and Li, Tzu-Mao and Maclaurin, Dougal and Ragan-Kelley, Jonathan},
  journal={arXiv preprint arXiv:2008.11256},
  year={2020}
}

@article{mazza2020automatic,
  author    = {Damiano Mazza and
               Michele Pagani},
  title     = {Automatic differentiation in {PCF}},
  journal   = {Proc. {ACM} Program. Lang.},
  volume    = {5},
  number    = {{POPL}},
  pages     = {1--27},
  year      = {2021},
  url       = {https://doi.org/10.1145/3434309},
  doi       = {10.1145/3434309},
  bibsource = {dblp computer science bibliography, https://dblp.org}
}

@inproceedings{zhu2020principles,
  author    = {Shaopeng Zhu and
               Shih{-}Han Hung and
               Shouvanik Chakrabarti and
               Xiaodi Wu},
  OPTeditor    = {Alastair F. Donaldson and
               Emina Torlak},
  title     = {On the principles of differentiable quantum programming languages},
  booktitle = {Proceedings of the 41st {ACM} {SIGPLAN} International Conference on
               Programming Language Design and Implementation, {PLDI} 2020, London,
               UK, June 15-20, 2020},
  pages     = {272--285},
  publisher = {{ACM}},
  year      = {2020},
  url       = {https://doi.org/10.1145/3385412.3386011},
  doi       = {10.1145/3385412.3386011},
  bibsource = {dblp computer science bibliography, https://dblp.org}
}

@article{bettencourt2019taylor,
  title={Taylor-Mode Automatic Differentiation for Higher-Order Derivatives in {J}{A}{X}},
  author={Bettencourt, Jesse and Johnson, Matthew J and Duvenaud, David},
  year={2019}
}

@article{frostig2018compiling,
  title={Compiling machine learning programs via high-level tracing},
  author={Frostig, Roy and Johnson, Matthew James and Leary, Chris},
  journal={Systems for Machine Learning},
  year={2018}
}

@article{paszke2017automatic,
  title={Automatic differentiation in pytorch},
  author={Paszke, Adam and Gross, Sam and Chintala, Soumith and Chanan, Gregory and Yang, Edward and DeVito, Zachary and Lin, Zeming and Desmaison, Alban and Antiga, Luca and Lerer, Adam},
  year={2017}
}

@inproceedings{laue2020simple,
  title={A Simple and Efficient Tensor Calculus.},
  author={Laue, S{\"o}ren and Mitterreiter, Matthias and Giesen, Joachim},
  booktitle={AAAI},
  pages={4527--4534},
  year={2020}
}

@article{laue2018computing,
  title={Computing higher order derivatives of matrix and tensor expressions},
  author={Laue, S{\"o}ren and Mitterreiter, Matthias and Giesen, Joachim},
  journal={Advances in Neural Information Processing Systems},
  volume={31},
  pages={2750--2759},
  year={2018}
}

@inproceedings{chen2018neural,
  title={Neural ordinary differential equations},
  author={Chen, Ricky TQ and Rubanova, Yulia and Bettencourt, Jesse and Duvenaud, David K},
  booktitle={Advances in Neural Information Processing Systems},
  pages={6571--6583},
  year={2018}
}

@article{knoll2004jacobian,
  title={Jacobian-free {N}ewton--{K}rylov methods: a survey of approaches and applications},
  author={Knoll, Dana A and Keyes, David E},
  journal={Journal of Computational Physics},
  volume={193},
  number={2},
  pages={357--397},
  year={2004},
  publisher={Elsevier}
}

@book{amari2012differential,
  title={Differential-geometrical methods in statistics},
  author={Amari, Shun-ichi},
  volume={28},
  year={2012},
  publisher={Springer Science \& Business Media}
}

@incollection{jacobsrutten2011,
  title={An introduction to (co)algebras and (co)induction},
  author={Bart Jacobs and JMMM Rutten},
  booktitle={Advanced Topics in Bisimulation and Coinduction},
  pages={38--99},
  year={2011},
  publisher={CUP}
}

@techreport{bendtsen1996fadbad,
  title={FADBAD, a flexible {C}++ package for automatic differentiation},
  author={Bendtsen, Claus and Stauning, Ole},
  year={1996},
  institution={Technical Report IMM--REP--1996--17, Department of Mathematical Modelling, Technical University of Denmark, Lyngby}
}

@article{bendtsen1997tadiff,
  title={TADIFF, a flexible C++ package for automatic differentiation},
  author={Bendtsen, Claus and Stauning, Ole},
  journal={TU of Denmark, Department of Mathematical Modelling, Lungby. Technical report IMM-REP-1997-07},
  year={1997}
}

@article{karczmarczuk2001functional,
  title={Functional differentiation of computer programs},
  author={Karczmarczuk, Jerzy},
  journal={Higher-Order and Symbolic Computation},
  volume={14},
  number={1},
  pages={35--57},
  year={2001},
  publisher={Springer}
}

@article{pearlmutter2007lazy,
  title={Lazy multivariate higher-order forward-mode AD},
  author={Pearlmutter, Barak A and Siskind, Jeffrey Mark},
  journal={ACM SIGPLAN Notices},
  volume={42},
  number={1},
  pages={155--160},
  year={2007},
  publisher={ACM New York, NY, USA}
}

@article{wang2016capitalizing,
  title={Capitalizing on live variables: new algorithms for efficient Hessian computation via automatic differentiation},
  author={Wang, Mu and Gebremedhin, Assefaw and Pothen, Alex},
  journal={Mathematical Programming Computation},
  volume={8},
  number={4},
  pages={393--433},
  year={2016},
  publisher={Springer}
}

@book{gelfand2000calculus,
  title={Calculus of variations},
  author={Gelfand, Izrail Moiseevitch and Silverman, Richard A and Silverman, Richard A},
  year={2000},
  publisher={Courier Corporation}
}

@article{cockett2011faa,
  title={The {F}aa di {B}runo construction},
  author={Cockett, J Robin B and Seely, Robert AG},
  journal={Theory and Applications of Categories},
  volume={25},
  number={15},
  pages={394--425},
  year={2011}
}

@inproceedings{vakar2021reverse,
  title={Reverse {A}{D} at Higher Types: Pure, Principled and Denotationally Correct.},
  author={V{\'a}k{\'a}r, Matthijs},
  booktitle={ESOP},
  pages={607--634},
  year={2021}
}

@article{vakar2020denotational,
  title={Denotational Correctness of Foward-Mode Automatic Differentiation for Iteration and Recursion},
  author={V{\'a}k{\'a}r, Matthijs},
  journal={arXiv preprint arXiv:2007.05282},
  year={2020}
}

@article{vakar2021chad,
  title={{CHAD}: Combinatory Homomorphic Automatic Differentiation},
  author={V{\'a}k{\'a}r, Matthijs and Smeding, Tom},
  journal={arXiv preprint arXiv:2103.15776},
  year={2021}
}

@article{lucatelli2021chad,
  title={{CHAD} for Expressive Total Languages},
  author={Lucatelli Nunes, Fernando and V{\'a}k{\'a}r, Matthijs},
  journal={arXiv e-prints},
  pages={arXiv--2110},
  year={2021}
}
\clearpage
\appendix
\section{$\CartSp$ and $\Man$ are not cartesian closed categories}
\label{sec:man_not_ccc}
\begin{lem}\label{lem:borsuk}
There is no continuous injection $\RR^{d+1}\to \RR^d$.
\end{lem}
\begin{proof}
If there were, it would restrict to a continuous injection
$S^d\to \RR^d$.
The Borsuk-Ulam theorem, however, tells us that every continuous 
$f:S^d\to\RR^d$ has some $x\in S^d$ such that $f(x)=f(-x)$, which is a
contradiction.
\end{proof}
For $n\geq1$, define the terms:
\[
\var_1:\reals,\ldots, \var_n:\reals\vdash \trm_n=
\fun{\var[2]}{\var_1 * \var[2] + \dots + \var_n*\var[2]^n}:\reals\To\reals
\]
Assuming that $\CartSp$ or $\Man$ is cartesian closed, observe that these get
interpreted as injective continuous functions (because they are smooth) $\RR^n\to \sem{\reals\To\reals}$ in $\CartSp$ and $\Man$.
\begin{thm}
$\CartSp$ is not cartesian closed.
\end{thm}
\begin{proof}
If $\CartSp$ were cartesian closed, we would have $\sem{\reals\To\reals}=\RR^n$
for some $n$.
Then, in particular, we would get a continuous 
injection $\sem{\trm_{n+1}}:\RR^{n+1}\to \RR^n$,
which contradicts Lemma \ref{lem:borsuk}.
\end{proof}
\begin{thm}
$\Man$ is not cartesian closed.
\end{thm}
\begin{proof}
Observe that we have $\iota_n:\RR^n\to \RR^{n+1};\,\tTuple{a_1,\ldots,a_n}\mapsto \tTuple{a_1,\ldots,a_n,0}$
and that $\iota_n;\sem{\trm_{n+1}}=\sem{\trm_n}$.
Let us write $A_n$ for the image of $\sem{\trm_n}$ and $A=\cup_{n\geq1}A_n$.
Then $A_n$ is connected because it is the continuous image of a connected set.
Similarly, $A$ is connected because it is the non-disjoint union of connected sets.
This means that $A$ lies in a single connected component of $\sem{\reals\To\reals}$,
which is a manifold with some finite dimension, say $d$.

Choose some $x\in\RR^{d+1}$ (say, $0$), some open $d$-ball $U$ around $\sem{\trm_{d+1}}(x)$,
and some open $d+1$-ball $V$ around $x$ in $\sem{\trm_{d+1}}^{-1}(U)$.
Then $\sem{\trm_{d+1}}$ restricts to a continuous injection from $V$ to $U$, or equivalently,
from $\RR^{d+1}$ to $\RR^d$, which contradicts Lemma \ref{lem:borsuk}.
\end{proof}

\end{document}